\documentclass[letterpaper, 11pt]{article}

\usepackage{amsmath, amssymb, amsfonts, amsthm}
\usepackage{bm}
\usepackage{mathtools}
\usepackage{mathrsfs}
\usepackage{geometry}
\usepackage{setspace}
\usepackage[utf8x]{inputenc}
\usepackage{algpseudocode,algorithmicx,algorithm}
\usepackage[T1]{fontenc}
\usepackage{lmodern}
\usepackage{graphicx}
\usepackage{enumerate}
\usepackage[natural, dvipsnames, pdftex]{xcolor}
\usepackage{tikz}
\usepackage{verbatim}
\usepackage{hyperref}
\usepackage{cancel}
\usepackage[all]{xy}
\usepackage{mdframed}
\usepackage[capitalize]{cleveref}
\usepackage{todonotes}
\usepackage{algpseudocode}
\usepackage{comment}
\usepackage{booktabs}
\usepackage{tabularx}
\usepackage{array}
\usepackage{float}

\usepackage{float}
\usepackage{xcolor}
\usepackage{colortbl}

\definecolor{pqhead}{RGB}{232,239,246}
\definecolor{pqrule}{RGB}{157,174,190}

\usepackage[numbers]{natbib}

\usepackage{enumitem}

\usepackage{indentfirst}
\theoremstyle{plain}
\newtheorem{thm}{Theorem}[section]
\newtheorem{lemma}[thm]{Lemma}
\newtheorem{prop}[thm]{Proposition}
\newtheorem{cor}[thm]{Corollary}

\theoremstyle{definition}
\newtheorem{defn}[thm]{Definition}
\newtheorem{conj}[thm]{Conjecture}
\newtheorem{ex}[thm]{Example}

\theoremstyle{remark}
\newtheorem{rmk}[thm]{Remark}

\newcommand{\F}{\mathbb{F}}

\newcommand{\cC}{\mathcal{C}}

\begin{document}

\title{List Size Bounds for Random Linear Codes: List Decoding and Average Radius List Decoding}

\author{
Chen Yuan$^{*}$ \qquad Ruiqi Zhu$^{*}$\\[0.8em]
\normalsize $^{*}$Shanghai Jiao Tong University, Shanghai, China\\[0.3em]
\normalsize \texttt{\{chen\_yuan,sjtuzrq7777\}@sjtu.edu.cn}
}
\date{}
\maketitle

\begin{abstract}
Guruswami, Li, Mosheiff, Resch, Silas and Wootters conjectured that, for every fixed prime power $q$ and $p\in(0,1-1/q)$, random linear codes of rate $1-h_q(p)-\varepsilon$ have list size $\frac{h_q(p)}{\varepsilon}\cdot(1+o(1))$ as $\varepsilon\to0$, both for list decoding and average-radius list decoding. The conjecture was known for $q=2$ and remained open for larger alphabets.

We prove the conjecture for list decoding for every prime power $q$, with the sharper bound $L=\frac{h_q(p)}{\varepsilon}+O_{p,q}(1).$ In contrast, we disprove the conjecture for average-radius list decoding. We show that the required list-size leading constant is at least $p\left(1+\log_q\frac{(q-1)(1-p)}{p}\right)$, which is strictly larger than $h_q(p)$ whenever $q^p(1-p)>1$. Moreover, for $p=1-\lambda/q$ with fixed $\lambda>1$, the required leading constant is at least $(\frac{\ln\lambda}{e\lambda}+o_q(1))\frac{q}{\ln q}$. For every fixed $q$ and $0<p<1-1/q$, we also prove that the required leading constant for average radius list decoding is at most $\frac{(q-1)h_q(p)}{\ln q}$. Thus, the supremum over $0<p<1-1/q$ of the required leading constant is $\Theta(q/\ln q)$ as $q\to\infty$.
\end{abstract}

\section{Introduction}
List decoding, introduced by~\cite{Elias1957,Wozencraft1958}, relaxes unique decoding by allowing the decoder to output a list of candidate codewords. A code $\cC\subseteq\mathbb{F}_q^n$ is $(p,L)$ list decodable if every Hamming ball of radius $pn$ contains fewer than $L$ codewords of $\cC$. For fixed $q$ and $0<p<1-1/q$, the list decoding capacity is $1-h_q(p)$ where
\[
h_q(p):=p\log_q(q-1)-p\log_q p-(1-p)\log_q(1-p).
\]
At rate $1-h_q(p)-\varepsilon$, random nonlinear codes achieve list size $O(1/\varepsilon)$ with high probability~\cite{Elias1991,Rudra2009}. This dependence on $\varepsilon$ is optimal for random codes~\cite{Rudra2009}. More precisely, the sharp leading constant is $1$: the classical random-coding upper bound gives list size $(1+o(1))/\varepsilon$~\cite{ZyablovPinsker1981}, and Kim et al.~\cite{KLT23} proved the matching lower bound. We study the precise list size needed by random linear codes at the same rate.

Linear codes have a concise description, but the linear dependencies among their codewords make the analysis more difficult. Zyablov and Pinsker~\cite{ZyablovPinsker1981} proved a $q^{O(1/\varepsilon)}$ list size bound for random linear codes. For the binary alphabet, Guruswami et al.~\cite{GuruswamiHastadSudanZuckerman2002} proved the existence of linear codes with list size $O(1/\varepsilon)$ at rates approaching capacity. Guruswami et al.~\cite{guruswami2010list} subsequently showed that, for every fixed prime power $q$, a random linear code of rate $1-h_q(p)-\varepsilon$ is $(p,C_{p,q}/\varepsilon)$ list decodable with high probability, for a constant $C_{p,q}$ depending only on $p,q$. Guruswami and Narayanan~\cite{GuruswamiNarayanan2014} proved that random linear codes also require list size at least $\Omega_{p,q}(1/\varepsilon)$. Thus, the list size was known to be of order $1/\varepsilon$, but the sharp constant remained open.

When the field size $q$ is allowed to grow as $\varepsilon\to0$, Alrabiah et al.~\cite{AlrabiahEtAl2025} proved that a random linear code of rate $R$ is $(1-R-\varepsilon,O(1/\varepsilon))$ list decodable with high probability over a field of size $q=2^{O(1/\varepsilon^2)}$. Goyal and Guruswami~\cite{GG26} reduced the required field size to $q=2^{O(\log(1/\varepsilon)/\varepsilon)}$, while keeping the same decoding radius and list size bound. Both results also hold for average radius list decoding. Our focus here is different: we fix $p$ and $q$ and let $\varepsilon\to0$.

\paragraph{The conjectured list size.}
For binary random linear codes, Li and Wootters~\cite{LiWootters2021} obtained the upper bound $h_2(p)/\varepsilon+O(1)$. Guruswami et al.~\cite{GuruswamiEtAl2022} proved a matching lower bound $h_q(p)/\varepsilon-O(1)$ over every finite field. This leads to the natural conjecture that the leading constant is $h_q(p)$ for every fixed $q$.

We next consider the stronger notion of average radius list decoding. A code is $(p,L)$ average radius list decodable if, for every set $\Lambda\subseteq \cC$ of $L$ distinct codewords and every center $\mathbf{z}\in\mathbb{F}_q^n$, we have $\frac{1}{L}\sum_{\mathbf{c}\in\Lambda} d(\mathbf{c},\mathbf{z})>pn$. Clearly, $(p,L)$ average radius list decodability implies $(p,L)$ list decodability. Rudra and Wootters~\cite{RW18} showed that, for constant $q$ and a range of fixed radii $p$, random linear codes of rate $1-h_q(p)-\varepsilon$ are average radius list decodable with list size $\exp(O_q(\log^3(1/\varepsilon)))$. For $q=2$, Guruswami et al.~\cite{GuruswamiEtAl2022} proved that the much sharper bound $h_2(p)/\varepsilon+O(1)$ holds for average radius list decoding. Motivated by these results, they conjectured that the leading constant $h_q(p)$ is correct for both notions over every fixed finite field~\cite[Section~1.4]{GuruswamiEtAl2022}.

\begin{conj}[Guruswami et al.~\cite{GuruswamiEtAl2022}]
\label{conj:list size}
Fix a prime power $q$ and $0<p<1-1/q$. A random linear code of rate $1-h_q(p)-\varepsilon$ is, with high probability, both $(p,L)$ list decodable and $(p,L)$ average radius list decodable for $L=\frac{h_q(p)}{\varepsilon}(1+o(1))$ where $o(1)\to0$ as $\varepsilon\to0$.
\end{conj}

For $q>2$, the list decoding bound of Guruswami et al.~\cite{guruswami2010list} does not directly give the same bound for average radius list decoding. Using the fact that a $(p,L)$ list decodable code is $(p-\gamma,O(L/\gamma))$ average radius list decodable~\cite[Appendix~E]{RudraWootters2015}, one obtains only $O_{p,q}(1/\varepsilon^2)$ list size at rate $1-h_q(p)-\varepsilon$. Recently, Guruswami et al.~\cite{GuruswamiLiSinghal2026} proved a list size bound of $C_{p,q}/\varepsilon$ for every fixed $q$ and $0<p<1-1/q$. From their proof, one may take\footnote{This follows from the choice of parameters in~\cite[Section~3.3 and Remark~4]{GuruswamiLiSinghal2026} and the proof of~\cite[Theorem~10]{guruswami2010list}.}
\[
C_{p,q}=\exp\!\left(O\!\left(\frac{(1-p)(\ln q)^2}{p(1-1/q-p)^2}\right)\right).
\]
They asked whether $C_{p,q}$ can be replaced by an absolute constant, or even by a constant close to $h_q(p)$.

We prove Conjecture~\ref{conj:list size} for list decoding, but disprove it for average radius list decoding. For list decoding, we establish the stronger upper bound $h_q(p)/\varepsilon+O_{p,q}(1)$ over every fixed finite field. \footnote{In independent and concurrent work, Silas~\cite{Silas2026} also proves a list size bound of $h_q(p)/\varepsilon+O_{p,q}(1)$. Kumar and Mon~\cite{KM26} concurrently obtain a uniform $O(qh_q(p)/\varepsilon)$ bound over the full range $0<p<1-1/q$.} For average radius list decoding, we give an explicit lower bound whose leading constant exceeds $h_q(p)$ in a nonempty range of radii for every $q\ge 3$. We further show that, for $p=1-\lambda/q$ with any fixed $\lambda>1$, the leading constant in our lower bound is $\left(\frac{\ln\lambda}{e\lambda}+o_q(1)\right)q/\ln q$. Thus, not only can the leading constant exceed $h_q(p)$, but there is also no absolute constant $C$ for which a $C/\varepsilon$ average radius list size bound holds over all finite fields. In particular, while list decoding and average radius list decoding have the same asymptotically optimal list size over the binary alphabet, their leading constant can differ over larger alphabets.

For every fixed $q$ and $0<p<1-1/q$, we also prove that a random linear code of rate $1-h_q(p)-\varepsilon$ is $(p,L)$ average radius list decodable with high probability, where $L=(q-1)h_q(p)/(\varepsilon\ln q)+O_{p,q}(1)$ as $\varepsilon\to0$. Since $h_q(p)<1$, the leading constant in our upper bound is at most $(q-1)/\ln q$ for all such $p$. For $p=1-\lambda/q$ with fixed $\lambda>1$, the preceding estimate for $C_{p,q}$ is $\exp(O_\lambda(q(\ln q)^2))$, while our upper bound has leading constant $(1+o_q(1))q/\ln q$. Together with our lower bound, this shows that the supremum over $0<p<1-1/q$ of the required leading constant is $\Theta(q/\ln q)$ as $q\to\infty$.

\subsection{Overview of Our Approach and Contributions}
\label{sec:intro-techniques}

Throughout the results below, $p$ and $q$ are fixed, and $\varepsilon>0$ is taken sufficiently small. Once these parameters are fixed, $n$ is taken sufficiently large. In particular, the parameters $m$ and $s$ in the lower bound constructions do not depend on $n$.

\subsubsection{Upper bound for list decoding}
We briefly describe the main ideas behind our proof. Suppose that there is a bad list of size $t$. By Lemma~\ref{lem:affine-reformulation}, this is equivalent to the existence of a subspace $D\subseteq \cC$ and a vector $\mathbf{z}\in\mathbb{F}_q^n$ such that $D_z:=(\mathbf{z}+D)\cap B_q^n(\mathbf{0},p)$ of size $L\geq t$. Choose any $t$ distinct vectors $\mathbf{x}_1,\ldots,\mathbf{x}_t\in D_z$ with affine dimension $r$ and $\mathbf{x}_1,\ldots,\mathbf{x}_{r+1}$ form an affine basis. Define the affine space $\mathcal A_{r+1}:=\left\{(a_1,\ldots,a_{r+1})\in\mathbb F_q^{r+1}:\sum_{i=1}^{r+1}a_i=1\right\}.$ Then, for every $j\in[t]$, we can write $\mathbf{x}_j=\sum_{i=1}^{r+1}a_{j,i}\mathbf{x}_i$ where $\sum_{i=1}^{r+1}a_{j,i}=1.$  Let
\[
A:=
\{(a_{j,1},\ldots,a_{j,r+1}):j\in[t]\}
\subseteq \mathcal A_{r+1}.
\]
Then $|A|=t$ and $A$ contains the standard basis vectors $\mathbf{e}_1,\ldots,\mathbf{e}_{r+1}$.


The set $A$, together with $\mathbf{x}_1,\ldots,\mathbf{x}_{r+1}$, determines all the $t$ chosen vectors. Thus, we want to bound the size of the following set
$$
N(A):=
\left\{
(\mathbf{x}_1,\ldots,\mathbf{x}_{r+1})
\in B_q^n(\mathbf{0},p)^{r+1}:
\sum_{i=1}^{r+1}a_i\mathbf{x}_i\in B_q^n(\mathbf{0},p), \forall
(a_1,\ldots,a_{r+1})\in A
\right\}.
$$
A straightforward approach is to use only one additional constraint from $A\setminus\{\mathbf{e}_1,\ldots,\mathbf{e}_{r+1}\}$. Such a constraint has the form $\sum_{i=1}^{r+1} a_i \mathbf{x}_i\in B_q^n(0,p).$ A known bound shows that the probability that the sum of several independent vectors from $B_q^n(0,p)$ still lies in $B_q^n(0,p)$ is at most $q^{-\tau_{p,q}n}$ for some small constant $\tau_{p,q}<h_q(p)$. However, this exponent is not large enough for the bound on $|N(A)|$ that we need.

Our main technical contribution is a sharper bound on $|N(A)|$ for $|A|-(r+1)\geq q^{m+1}$ with some constant $m=O_{p,q,\eta}(1)$.  After an invertible change of variables and relabeling, we obtain $m$ constraints $\beta_j\mathbf{x}'_j+\mathbf{z}_j\in B_q^n(0,p)$ where $j=1,\ldots,m$, $\beta_j\neq 0$ and $\mathbf{z}_j$ depends only on $\mathbf{x}'_1,\ldots,\mathbf{x}'_{j-1},\mathbf{x}'_{m+1},\ldots,\mathbf{x}'_{r+1}$. There are two cases. If there exists some $j\in[m]$ such that $\mathrm{wt}(\mathbf{z}_j)\le \tau n$, we obtain a direct counting bound. Otherwise, $\mathrm{wt}(\mathbf{z}_j)>\tau n$ for every $j\in[m]$. Since both $\mathbf{x}'_j$ and $\beta_j\mathbf{x}'_j+\mathbf{z}_j$ belong to $B_q^n(0,p)$, the number of possible choices for $\mathbf{x}'_j$ is at most $q^{(h_q(p)-c_\tau)n}$ by Lemma~\ref{lem:shifted-ball-intersection}. In summary, we manage to prove that $N(A)$ has size at most $q^{(rh_q(p)+\eta)n}$ for any fixed small constant $\eta>0$.

To finish the proof, we distinguish two cases according to the affine dimension $r$ of the bad list. If $r+1<t/C_0$, where $C_0$ is the constant in~\cite[Theorem 10]{guruswami2010list}, that theorem shows that such lists occur with very small probability. Otherwise, $r+1\ge t/C_0$ and there are at most $q^{(r+1)t}$ possible sets $A$. Moreover, $\mathbf{x}_i-\mathbf{x}_1$ for $i=2,\ldots,r+1$ are $r$ linearly independent vectors and thus all of them are contained in a random linear code $\cC$ of dimension $n(1-h_q(p)-\varepsilon)$ with probability $q^{-(h_q(p)+\varepsilon)rn}$. Take a union bound over all types of $A$ and we complete the proof of Theorem \ref{thm:intro-main}.

\begin{thm}
\label{thm:intro-main}
Fix a prime power $q$ and $0<p<1-1/q$. For every fixed $\varepsilon>0$ with $1-h_q(p)-\varepsilon>0,$ let $\cC\subseteq\mathbb{F}_q^n$ be a uniformly random linear code of rate $R=1-h_q(p)-\varepsilon.$ Then, with probability at least $1-q^{-\Omega_{p,q}(\varepsilon n)},$ the code $\cC$ is $\left(p,\,\left\lceil\frac{h_q(p)}{\varepsilon}\right\rceil+O_{p,q}(1) \right)$ list decodable.
\end{thm}

The additive term in Theorem~\ref{thm:intro-main} does not depend on $\varepsilon$. Together with the lower bound of~\cite{GuruswamiEtAl2022}, this determines the list size up to $O_{p,q}(1)$ and proves the list decoding part of Conjecture~\ref{conj:list size}.

\subsubsection{Lower bounds for average radius list decoding}
\label{sec:intro-average-techniques}
Following the lower bound argument of~\cite[Section~4]{GuruswamiEtAl2022}, we construct a distribution $\tau$ for the rows of an $n\times L$ matrix $M$. We will show that the columns of such a matrix form a bad average radius list, and that a random linear code contains such a matrix with high probability.

For a single coordinate, choose $\alpha$ uniformly from $\mathbb{F}_q$. Independently, let $E_1,\ldots,E_m\sim\mathrm{Bernoulli}_q(p_1)$ and $T\sim\mathrm{Bernoulli}_q(p_2)$, where both $p_1$ and $p_2$ are slightly larger than $p$ for sufficiently large $m$. If $T=0$, set $B=0$; otherwise choose $B$ uniformly from $\mathbb{F}_q$, independently of the other variables. The corresponding row is
\[
    \mathbf{X}=\left(\alpha+E_1,\ldots,\alpha+E_m, (\alpha+B+aT)_{a\in\mathbb{F}_q}\right).
\]
Definition~\ref{def:bad-distribution} gives the exact choices of $p_1$ and $p_2$.

For an $n\times L$ matrix with this row distribution, let $\alpha_i$ denote the value of $\alpha$ in the $i$-th row, and define $\mathbf{z}=(\alpha_1,\ldots,\alpha_n)\in\mathbb{F}_q^n.$ Each of the first $m$ columns differs from $\mathbf{z}$ on a $p_1$ fraction of the coordinates, so its distance from $z$ is slightly larger than $pn$. For the remaining $q$ columns, if $T_i=0$ in row $i$, all $q$ elements equal $\alpha_i$; if $T_i\ne0$, exactly $q-1$ of them differ from $\alpha_i$. Hence the sum of the distances of all $L=m+q$ columns from $\mathbf{z}$ is $\bigl(mp_1+(q-1)p_2\bigr)n=pLn,$ by our choice of $p_1$ and $p_2$. Thus the full list has average distance at most $pn$, even though the first $m$ codewords individually lie outside the radius $pn$ ball. Proposition~\ref{prop:tau-bad} verifies this statement, together with the distinctness of the columns.

It remains to show that a random linear code contains a matrix with this row distribution. Let $\mathbf{X}\sim\tau$. For a full rank matrix $A\in\mathbb{F}_q^{r\times L}$, write $A\tau$ for the distribution of $A\mathbf{X}$. By Theorem~\ref{thm:abundance}, it is enough to show, for every such $A$, that $H_q(A\tau)\geq (1-R+\eta)\dim(A\tau)$ where $\eta>0$. Here $A\tau$ denotes the distribution of $A\mathbf{X}$ for $\mathbf{X}\sim\tau$ and $\dim(A\tau)$ is its dimension in Definition~\ref{def:type}. 

To prove Lemma~\ref{lem:not-implicitly-rare}, we let
\[
    Z:=\alpha+B,\qquad U_i:=E_i-B\quad(1\le i\le m),
\]
and let $\mathbf{W}:=(Z,U_1,\ldots,U_m,T).$ The original row $\mathbf{X}$ and $\mathbf{W}$ can be obtained linearly from one another. Moreover, $Z$ is uniform and independent of $(U_1,\ldots,U_m,T)$, so we first bound the entropy of linear combinations of $(U_1,\ldots,U_m,T)$ and add $Z$ back at the end.

There are two sources of additional entropy in our construction. First, $p_1$ is chosen slightly larger than $p$ and hence $h_q(p_1)= h_q(p)+\frac{p h_q'(p)}{m}+O_{p,q}(m^{-2}).$ Thus the Bernoulli variables $E_i$ already contribute an additional $p h_q'(p)/m$ to the entropy per coordinate, to first order in $1/m$. The second gain comes from $B$. To see this, consider $(E_1-B,\ldots,E_d-B)$ for $1\le d\le m.$ If $T=0$, then $B=0$. If $T\ne0$, which occurs with probability $p_2$, then $B$ is uniform over $\mathbb{F}_q$ and independent of $E_1,\ldots,E_d$. Lemma~\ref{lem:uniform-common-shift} shows that, for any fixed $0<c<1$ and all sufficiently large $m$,
\[
    H_q(E_1-B,\ldots,E_d-B\mid T)
    \ge
    d\,h_q(p_1)+p_2c\,\frac{d}{m}.
\]
Since $p_2$ is close to $p$ and $c$ can be chosen arbitrarily close to $1$, this gives an additional entropy gain close to $p/m$ per coordinate.

For a general full rank matrix $A$, we use the same Gaussian elimination and entropy chain rule decomposition as in~\cite[Lemma~4.4]{GuruswamiEtAl2022}. This reduces the entropy estimate to terms of two types. The first consists of independent Bernoulli variables together with a common Bernoulli shift and is handled by Corollary~\ref{cor:ordinary-shift}, which follows from~\cite[Lemma~4.5]{GuruswamiEtAl2022}. The second contains the variable $B$; after conditioning on $T$, it is handled by Lemma~\ref{lem:uniform-common-shift} as above. Combining these estimates and then adding back the independent uniform variable $Z$, Lemma~\ref{lem:not-implicitly-rare} gives
\[
    H_q(A\tau)\ge
    \left(
        h_q(p)+\frac{p(1+h_q'(p))-\zeta}{m}
    \right)\dim(A\tau)
\]
for every full-rank $A$, for any fixed sufficiently small $\zeta>0$ and all sufficiently large $m$. Thus the two gains, $ \frac{p h_q'(p)}{m}$ and $\frac{p}{m},$ combine to give the coefficient $p(1+h_q'(p))$. Choosing $m$ so that $\frac{p(1+h_q'(p))-\zeta}{m}>\varepsilon$ allows us to apply Theorem~\ref{thm:abundance}, giving the following lower bound.

\begin{thm}
\label{thm:intro-average}
Fix a prime power $q$, $p\in(0,1-1/q)$, and $\delta\in(0,p(1+h_q'(p)))$. For every sufficiently small $\varepsilon>0$ and sufficiently large $n$, a random linear code over $\mathbb{F}_q$ of rate $1-h_q(p)-\varepsilon$ is not $\left(p,\left\lfloor\frac{p(1+h_q'(p))-\delta}{\varepsilon}\right\rfloor\right)$ average radius list decodable with probability $1-q^{-\Omega_{p,q,\delta,\varepsilon}(n)}$.
\end{thm}

The improvement over the conjectured leading constant is $p(1+h_q'(p))-h_q(p)=p+\log_q(1-p).$ This is positive whenever $q^p(1-p)>1$, disproving the average radius part of Conjecture~\ref{conj:list size}. For every $q\ge3$, this condition holds for all sufficiently small positive $p$; see Remark~\ref{rem:improvement-range}.

\subsubsection{Stronger lower bounds}
\label{sec:intro-multiline-techniques}

In Section~\ref{sec:average-multiline}, we extend the preceding construction by replacing $T$ with independent $q$-ary Bernoulli variables $T_1,\ldots,T_s$, each with nonzero probability $p_2$. These variables are independent of $\alpha,E_1,\ldots,E_m$. We keep the first $m$ entries $\alpha+E_1,\ldots,\alpha+E_m$. For each $j$, we include the entries $\alpha+B+aT_j$ for $a\in\mathbb{F}_q^*$, and include $\alpha+B$ only once. Thus the row has $L=m+1+s(q-1)$ entries. We set $B=0$ if any $T_j=0$; otherwise, we choose $B$ uniformly over $\mathbb{F}_q$, independently of $\alpha,E_1,\ldots,E_m$. With the choices of $p_1,p_2$ in Definition~\ref{ml:def:distribution}, the columns again have average distance at most $pn$ from the center vector $\mathbf{z}=(\alpha_1,\ldots,\alpha_n)$; see Proposition~\ref{ml:prop:bad}.

The entropy argument follows the preceding proof. After fixing $T_1,\ldots,T_s$, the variable $B$ is uniform exactly when all of them are nonzero. This event has probability $p_2^s$, instead of $p_2$ in the previous construction, so the same estimate from Lemma~\ref{lem:uniform-common-shift} now enters with a factor $p_2^s$. This also introduces one additional issue: the last entropy term may contain several linear combinations of $T_1,\ldots,T_s$ together with $B$. Lemma~\ref{ml:lem:last-block} shows that the entropy gain from $B$ is still preserved in the presence of these additional $T_j$-terms, and gives the required lower bound. Combining this lemma and the estimates in Section~\ref{sec:average}, Lemma~\ref{ml:lem:not-rare} shows that the previous bound for $H_q(A\tau)$ continues to hold, with $p(1+h_q'(p))$ replaced by
\[
K_s:=p h_q'(p)+p^s\left(
1+(s-1)\left(1-\frac1q\right)h_q'(p)
\right).
\]
Applying Theorem~\ref{thm:abundance} then gives the following result.

\begin{thm}
\label{thm:intro-multiline}
Fix a prime power $q$, $p\in(0,1-1/q)$, an integer $s\ge1$ and let $K_s$ be as defined above. For every fixed $\delta\in(0,K_s)$, every sufficiently small $\varepsilon>0$ and sufficiently large $n$, a random linear code over $\mathbb{F}_q$ of rate $1-h_q(p)-\varepsilon$ is not $\left(p,\left\lfloor\frac{K_s-\delta}{\varepsilon}\right\rfloor\right)$ average radius list decodable with probability $1-q^{-\Omega_{p,q,s,\delta,\varepsilon}(n)}$.
\end{thm}
This is Theorem~\ref{ml:thm:lower bound}, proved in Appendix~\ref{app:multiline}. The case $s=1$ recovers the previous bound. For fixed $\lambda>1$, take $p=1-\lambda/q$ and choose $s$ to maximize $K_s$. Then $\max_{s\ge1} K_s = \left(\frac{\ln\lambda}{e\lambda}+o_q(1)\right) \frac{q}{\ln q}.$ Thus the coefficient of $1/\varepsilon$ required for average radius list decoding is at least a positive constant times $q/\ln q$, whereas the list decoding coefficient $h_q(p)$ is always less than $1$. Therefore, there is no absolute constant $C$ such that $C/\varepsilon$ suffices for average radius list decoding over all finite fields. This answers the question of~\cite[Section~4]{GuruswamiLiSinghal2026} negatively.

\subsubsection{Upper bound for average radius list decoding}
\label{sec:intro-average-upper}

Let $k:=\lfloor(1-h_q(p)-\varepsilon)n\rfloor$ and $r:=n-k$. Choose $H\in\mathbb{F}_q^{r\times n}$ uniformly at random among matrices of rank $r$ and let $\cC:=\ker H$. Fix an integer $L\ge2$ independent of $n$. Suppose that there exist $L$ distinct codewords $\mathbf{c}_1,\ldots,\mathbf{c}_L\in\cC$ and a vector $\mathbf{y}\in\mathbb{F}_q^n$ such that $\sum_{i=1}^L d(\mathbf{c}_i,\mathbf{y})\le Lpn$. Let $\mathbf{x}_i:=\mathbf{c}_i-\mathbf{y}$ for each $i\in[L]$. Then $\mathbf{x}_1,\ldots,\mathbf{x}_L$ are distinct and satisfy $H\mathbf{x}_1=\cdots=H\mathbf{x}_L$ and $\sum_{i=1}^L\mathrm{wt}(\mathbf{x}_i)\le Lpn$. Let $\tau:=\mathrm{Bernoulli}_q(p)^n$. As shown in the proof of Lemma~\ref{lem:ar-reformulation}, the latter inequality is equivalent to $\prod_{i=1}^L\tau(\mathbf{x}_i)\ge q^{-Lnh_q(p)}$.

For each $\mathbf{x}\in\mathbb{F}_q^n$, define $f(\mathbf{x}):=(1+q^r\tau(\mathbf{x}))^{1/n}$. Note that $f(\mathbf{x})\ge1$, $f(\mathbf{x})\ge q^{r/n}\tau(\mathbf{x})^{1/n}$ and $\sum_{\mathbf{x}\in\mathbb{F}_q^n}f(\mathbf{x})^n=q^n+q^r$. By the AM-GM inequality, the existence of such a list implies that
\[
\max_{\mathbf{b}\in\mathbb{F}_q^r}\sum_{\mathbf{x}:H\mathbf{x}=\mathbf{b}}(f(\mathbf{x})-1)
\ge L(q^{(r-nh_q(p))/n}-1)\ge L(q^\varepsilon-1).
\]
Using the polynomial method in~\cite{KM26}, we prove in Lemma~\ref{lem:ar-sum} that the expectation of the left hand side over $H$ is at most $q-1+o(1)$. By Markov's inequality, $\cC$ is not $(p,L)$ average radius list decodable with probability at most $(q-1+o(1))/(L(q^\varepsilon-1))$.

To obtain a better bound on $L$, let $N:=r+\lfloor\sqrt n\rfloor$. Choose $A\in\mathbb{F}_q^{N\times n}$ and $B\in\mathbb{F}_q^{r\times N}$ independently and uniformly at random among matrices of rank $N$ and $r$, respectively, and set $H:=BA$. Then $H$ is uniformly distributed among matrices of rank $r$. For each $H$ for which $\cC=\ker H$ is not $(p,L)$ average radius list decodable, fix $L$ distinct vectors $\mathbf{x}_1,\ldots,\mathbf{x}_L$ satisfying $H\mathbf{x}_1=\cdots=H\mathbf{x}_L$ and $\prod_{i=1}^L\tau(\mathbf{x}_i)\ge q^{-Lnh_q(p)}$. Let $\mathbf{z}_i:=A\mathbf{x}_i$ for each $i\in[L]$. Given $H$, the kernel of $A$ is a uniformly random subspace of $\ker H$ of dimension $n-N$. Each difference $\mathbf{x}_i-\mathbf{x}_j$ with $i\ne j$ is a fixed nonzero vector in $\ker H$. Thus $\Pr[\mathbf{z}_i=\mathbf{z}_j\mid H]\le q^{r-N}$, and by a union bound the vectors $\mathbf{z}_1,\ldots,\mathbf{z}_L$ are distinct with conditional probability at least $1-\binom{L}{2}q^{r-N}=1-o(1)$.

Since $H=BA$, we have $B\mathbf{z}_1=\cdots=B\mathbf{z}_L$. Also, $(A\tau)(\mathbf{z}_i)\ge\tau(\mathbf{x}_i)$ for each $i\in[L]$. On $\mathbb{F}_q^N$, we now define $f(\mathbf{z}):=(1+q^r(A\tau)(\mathbf{z}))^{1/N}$. By the same argument using the AM-GM inequality, if $\mathbf{z}_1,\ldots,\mathbf{z}_L$ are distinct, then
\[
\max_{\mathbf{b}\in\mathbb{F}_q^r}\sum_{\mathbf{z}:B\mathbf{z}=\mathbf{b}}(f(\mathbf{z})-1)
\ge L(q^{(r-nh_q(p))/N}-1).
\]
For each fixed $A$, we have $f(\mathbf{z})\ge1$ for all $\mathbf{z}\in\mathbb{F}_q^N$ and $\sum_{\mathbf{z}\in\mathbb{F}_q^N}f(\mathbf{z})^N=q^N+q^r$. Applying Lemma~\ref{lem:ar-sum} and then taking the expectation over $A$ shows that the expectation of the left hand side over $A$ and $B$ is at most $q-1+o(1)$. Since $(r-nh_q(p))/N\to\varepsilon/(h_q(p)+\varepsilon)$ as $n\to\infty$ and the vectors $\mathbf{z}_1,\ldots,\mathbf{z}_L$ are distinct with conditional probability $1-o(1)$ whenever a bad list exists, the probability that $\cC$ is not $(p,L)$ average radius list decodable is at most $\frac{q-1}{L(q^{\varepsilon/(h_q(p)+\varepsilon)}-1)}+o(1)$; see Proposition~\ref{prop:ar-positive-probability}.

Suppose that $L(q^{\varepsilon/(h_q(p)+\varepsilon)}-1)>q-1$. By continuity, the same inequality holds with some constant $\varepsilon'\in(0,\varepsilon)$ in place of $\varepsilon$. Hence a random linear code of rate $1-h_q(p)-\varepsilon'$ is $(p,L)$ average radius list decodable with probability bounded below by a positive constant. Rossman's threshold theorem~\cite{Rossman2019} then implies that at rate $1-h_q(p)-\varepsilon$, this probability is at least $1-q^{-\Omega_{p,q,\varepsilon}(n)}$ for the choice of $L$ below. We obtain the following theorem.

\begin{thm}
\label{thm:intro-average-upper}
Fix a prime power $q$, $0<p<1-1/q$ and $0<\varepsilon<1-h_q(p)$. For all sufficiently large $n$, a random linear code $\cC\subseteq\mathbb{F}_q^n$ of dimension $\lfloor(1-h_q(p)-\varepsilon)n\rfloor$ is $(p,L)$ average radius list decodable with probability at least $1-q^{-\Omega_{p,q,\varepsilon}(n)}$, where
\[
L=\left\lfloor\frac{q-1}{q^{\varepsilon/(h_q(p)+\varepsilon)}-1}\right\rfloor+1.
\]
In particular, for fixed $p$ and $q$, we have $L=\frac{(q-1)h_q(p)}{\varepsilon\ln q}+O_{p,q}(1)$ as $\varepsilon\to0$.
\end{thm}

This is Theorem~\ref{thm:ar-upper}, proved in Section~\ref{sec:average-upper}.

\paragraph{Organization.}
Section~\ref{sec:prelims} collects the notation, definitions and results used in the proofs. Section~\ref{sec:RLC-list-size} proves the list decoding upper bound. Section~\ref{sec:average} proves our first lower bound for average radius list decoding, and Section~\ref{sec:average-multiline} strengthens this lower bound. Section~\ref{sec:average-upper} proves the upper bound for average radius list decoding. The auxiliary entropy estimates are proved in Appendices~\ref{app:proof-glmrsw-lemma45} and~\ref{app:multiline}.

\section{Preliminaries}\label{sec:prelims}
We collect the notation and facts used in the proofs. List decoding and average-radius list decoding were defined in the Introduction. A random linear code $\cC\subseteq \F_q^n$ of dimension $k$ is a uniformly random $k$ dimensional subspace of $\mathbb{F}_q^n$ and the rate is defined as $R=k/n$. Equivalently, a random linear code of dimension $k$ is the kernel of a matrix $H\in\mathbb{F}_q^{(n-k)\times n}$ chosen uniformly at random among matrices of rank $n-k$.

\subsection{Notation}

Throughout this paper, $q$ is a prime power, and all spans, ranks and dimensions are over $\mathbb{F}_q$ unless stated otherwise. For every positive integer $t$, we view $\mathbb{F}_{q^t}$ as an extension of $\mathbb{F}_q$ of degree $t$. We write $\ln$ for the natural logarithm. Write $[m]:=\{1,\ldots,m\}$ and $\mathbb{F}_q^*:=\mathbb{F}_q\setminus\{0\}$. For $\mathbf{x}\in\mathbb{F}_q^m$ and $I\subseteq[m]$, let $\mathbf{x}_I$ be the vector of coordinates indexed by $I$ and write $(\mathbf{x})_i=x_i$. We use $\mathbf{1}_m$ for the all one vector and $\mathbf{e}_i$ for the $i$-th standard basis vector.

For $\mathbf{x},\mathbf{y}\in\mathbb{F}_q^n$, define $\mathrm{supp}(\mathbf{x}):=\{i:x_i\ne0\}$, $\mathrm{wt}(\mathbf{x}):=|\mathrm{supp}(\mathbf{x})|$ and $d(\mathbf{x},\mathbf{y}):=\mathrm{wt}(\mathbf{x}-\mathbf{y})$. For $0\le p\le1$, write $B_q^n(\mathbf{y},p):=\{\mathbf{x}\in\mathbb{F}_q^n:d(\mathbf{x},\mathbf{y})\le pn\}.$ Recall the $q$-ary entropy function
\[
h_q(p):=p\log_q(q-1)-p\log_q p-(1-p)\log_q(1-p),
\qquad 0\le p\le1.
\]
For fixed $q$ and $0<p<1-1/q$, we have $q^{h_q(p)n-O_{p,q}(\log n)}\le |B_q^n(\mathbf{y},p)|\le q^{h_q(p)n}.$

Define $\mathcal{A}_{r+1}:=\{\mathbf{a}\in\mathbb{F}_q^{r+1}:\sum_{i=1}^{r+1} a_i=1\}$ be an affine space of dimension $r$. Define $\mathbf{e}_1,\ldots,\mathbf{e}_{r+1}\in \F_q^{r+1}$ to be the standard basis vector. $\mathbf{x}_1,\ldots,\mathbf{x}_s\in\mathbb{F}_q^n$ are said to be \emph{affinely independent} if $\mathbf{x}_2-\mathbf{x}_1,\ldots,\mathbf{x}_s-\mathbf{x}_1$ are linearly independent. The affine dimension of $\mathbf{x}_1,\ldots,\mathbf{x}_s$ is $\dim(\mathrm{span}\{\mathbf{x}_2-\mathbf{x}_1,\ldots,\mathbf{x}_s-\mathbf{x}_1\}).$ In particular, a set of affine dimension $r$ contains $r+1$ $\F_q$-affinely independent vectors, which we call an affine basis.

\subsection{Entropy notation}

For a random variable $X$ on a finite set $\mathcal{X}$, define
\[
H(X):=-\sum_{x\in\mathcal{X}}\Pr[X=x]\log\Pr[X=x],
\qquad H_q(X):=\frac{H(X)}{\log q}.
\]
For a distribution $\tau$ on a finite set $\mathcal{X}$, let $X\sim\tau$ and write $\tau(x):=\Pr[X=x]$ for each $x\in\mathcal{X}$. Define $H(\tau):=H(X)$, $H_q(\tau):=H_q(X)$ and $\mathrm{supp}(\tau):=\{x\in\mathcal{X}:\tau(x)>0\}$. We write $X\sim \mathrm{Unif}(S)$ when $X$ is uniform on a finite set $S$. For random variables $X$ and $Y$, we write $X\perp Y$ when they are independent.

For a distribution $\tau$ on $\mathbb{F}_q^L$ and a matrix $A\in\mathbb{F}_q^{L'\times L}$, let $A\tau$ denote the distribution of $A\mathbf{X}$, where $\mathbf{X}\sim\tau$ is viewed as a column vector. For each $\mathbf{z}\in\mathbb{F}_q^{L'}$, we have $(A\tau)(\mathbf{z})=\Pr[A\mathbf{X}=\mathbf{z}]=\sum_{\mathbf{x}\in\mathbb{F}_q^L:\,A\mathbf{x}=\mathbf{z}}\tau(\mathbf{x})$.

For $s\in[0,1-1/q]$, a random variable $X\sim\mathrm{Bernoulli}_q(s)$ satisfies
\[
\Pr[X=0]=1-s,\qquad
\Pr[X=a]=\frac{s}{q-1}\quad(a\in\mathbb{F}_q^*).
\]
Then $H_q(X)=h_q(s)$. The notation $\mathrm{Bernoulli}_q(s)^d$ denotes the joint distribution of $d$ independent $\mathrm{Bernoulli}_q(s)$ random variables. We write $\mathbf{1}[\mathcal E]$ for $1$ when the event $\mathcal E$ holds and $0$ otherwise.

We use the following entropy identities and inequalities; see~\cite[Chapter~2]{CoverThomas2006}. Write $H(X\mid Y)$ for conditional entropy. The chain rule and the definitions of mutual information give $H(X,Y)=H(Y)+H(X\mid Y),\,I(X;Y):=H(X)-H(X\mid Y)$ and $I(X;Y\mid Z):=H(X\mid Z)-H(X\mid Y,Z).$ A subscript $q$ indicates normalization by $\log q$. For any function $f$ of $Y$, the data-processing inequalities give $H(X\mid f(Y))\ge H(X\mid Y),\,I(X;Y)\ge I(X;f(Y))$ and $I(X;Y\mid Z)\ge I(X;f(Y)\mid Z).$ If $f(Y)$ is an estimate of $X$ and $|\mathcal{X}|\ge2$, let $p_{\mathrm{err}}:=\Pr[f(Y)\ne X]$. Fano's inequality states that $H(X\mid Y)\le h_2(p_{\mathrm{err}})+p_{\mathrm{err}}\log(|\mathcal{X}|-1).$

We also use Hoeffding's inequality: for independent random variables $Y_1,\ldots,Y_N$ taking values in $[0,1]$ and $t>0$,
\[
\Pr\!\left[\sum_{i=1}^N(Y_i-\mathbb{E}Y_i)\ge t\right]
\le \exp(-2t^2/N).
\]
The same bound holds for the event $\sum_{i=1}^N(Y_i-\mathbb{E}Y_i)\le-t$.

\subsection{Types and implicit rarity}

We recall the definitions and theorem from~\cite{GuruswamiEtAl2022} used for our lower bounds. For $M\in\mathbb{F}_q^{n\times L}$, we view its columns as candidate codewords and write $M\subseteq\cC$ when every column belongs to $\cC$.

\begin{defn}[Types~\cite{GuruswamiEtAl2022}]
\label{def:type}
The type of $M\in\mathbb{F}_q^{n\times L}$ is its empirical row distribution $\tau_M$ on $\mathbb{F}_q^L$, given by
\[
\tau_M(\mathbf{v})=
\frac{|\{i\in[n]: \text{the }i\text{'th row of }M\text{ is equal to }\mathbf{v}\}|}{n}.
\]
For a distribution $\tau$ on $\mathbb{F}_q^L$, let $\mathcal{M}_{n,\tau}$ be the set of matrices with type $\tau$ and set $\dim(\tau):=\dim\mathrm{span}(\mathrm{supp}(\tau))$. 
\end{defn}

\begin{rmk}[{\cite[Remark~2.3]{GuruswamiEtAl2022}}]
\label{rem:type-rounding}
For some distributions $\tau$ over $\mathbb{F}_q^L$, the set $\mathcal{M}_{n,\tau}$ may be empty because $n\tau(\mathbf{v})$ is not an integer. For such $\tau$, one can use matrices with either $\lfloor n\tau(\mathbf{v})\rfloor$ or $\lceil n\tau(\mathbf{v})\rceil$ copies of each row $\mathbf{v}$. This has a negligible effect on the analysis, as we always take $n$ sufficiently large compared to the other parameters. We therefore omit these rounding details in the proofs.

For our average-radius constructions in Section~\ref{sec:average} and~\ref{sec:average-multiline}, we instead round down the number of each nonzero row and fill the remaining positions with zero rows. The zero row has positive probability in these constructions, and this rounding does not increase the average distance.
\end{rmk}

For $M\in\mathcal{M}_{n,\tau}$, we have $\mathrm{rank}(M)=\dim(\tau)$. Whenever $\mathcal M_{n,\tau}$ is nonempty, the number of matrices with type $\tau$ satisfies $|\mathcal M_{n,\tau}|= q^{nH_q(\tau)+O_{q,L}(\log n)}.$

If $M$ has type $\tau$ and $A\in\mathbb{F}_q^{L'\times L}$, then $MA^{\mathsf T}$ has type $A\tau$. Moreover, $M\subseteq\cC$ implies $MA^{\mathsf T}\subseteq\cC$ when $\cC$ is linear. Thus we must also consider the distributions obtained by taking linear combinations of the columns.
Here $\dim(A\tau)$ is defined as in Definition~\ref{def:type}; it equals $\mathrm{rank}(A)$ if $\mathrm{supp}(\tau)$ spans $\mathbb{F}_q^L$, but may be smaller otherwise.

\begin{defn}[$\gamma$-implicitly rare~\cite{GuruswamiEtAl2022}]\label{def:implicitly-rare}
We say that a distribution $\tau$ over $\mathbb{F}_q^L$ is $\gamma$-implicitly rare if there exists a full-rank linear transformation $A:\mathbb{F}_q^L \to \mathbb{F}_q^{L'}$ where $L' \le L$ such that $H_q(A\tau) < \gamma \cdot \dim(A\tau).$
\end{defn}

Thus $\tau$ is not $\gamma$-implicitly rare if and only if $H_q(A\tau)\ge\gamma\,\dim(A\tau)$ for every such $A$. The following theorem is due to Mosheiff et al.~\cite{mosheiff2020ldpc}; we use the formulation in~\cite[Theorem~2.5]{GuruswamiEtAl2022}.

\begin{thm}[{\cite[Theorem~2.5]{GuruswamiEtAl2022}}]
\label{thm:abundance}
Let $R\in(0,1)$ and fix $\eta>0$. Let $\tau$ be a $(1-R-\eta)$-implicitly rare distribution over $\mathbb{F}_q^L$ $(L\in\mathbb{N})$, and let $\cC$ be a random linear code of rate $R$. Then
\[
\Pr\!\left[\exists M\in\mathcal{M}_{n,\tau}: M\subseteq \cC\right]
\le q^{-\eta n}.
\]
Conversely, suppose that $\tau$ is not even $(1-R+\eta)$-implicitly rare. Then
\[
\Pr\!\left[\exists M\in\mathcal{M}_{n,\tau}: M\subseteq \cC\right]
\ge 1-n^{O_{L,q}(1)}\cdot q^{-\eta n}.
\]
\end{thm}

\section{Asymptotically Optimal List Size for List Decoding}\label{sec:RLC-list-size}
Fix $0<p<1-\frac1q$ and $\varepsilon>0$. Let $\cC\subseteq \mathbb{F}_q^n$ be a uniformly random $k$ dimensional linear code of rate $R=\frac{k}{n} = 1-h_q(p)-\varepsilon.$ Our goal is to show that there exists a constant $c_{p,q}>0$, depending only on $p$ and $q$, such that with high probability,
\[
\max_{\mathbf{y}\in\mathbb{F}_q^n}
\left|
\cC\cap B_q^n(\mathbf{y},p)
\right|
\le
\left\lceil
\frac{h_q(p)}{\varepsilon}
\right\rceil
+c_{p,q}.
\]
Equivalently, with high probability, $\cC$ is
\[
\left(
p,\,
\frac{h_q(p)}{\varepsilon}+O_{p,q}(1)
\right)
\text{ list decodable}.
\]

We now reformulate the failure of list decodability. Specifically, we show that this is equivalent to the existence of $D\subseteq \cC$ and $\mathbf{z}\in\mathbb{F}_q^n$ such that $\mathbf{z}+D$ contains many vectors of $B_q^n(\mathbf{0},p)$.

\begin{lemma}
\label{lem:affine-reformulation}
Let $\cC\subseteq \mathbb{F}_q^n$ be a linear code and let $t\ge 2$ be an integer. If there exists $\mathbf{y}\in\mathbb{F}_q^n$ such that $B_q^n(\mathbf{y},p)\cap \cC$ has size $t$ and affine dimension $r$, then there exists a subspace $D\subseteq \cC$ of dimension $r$ and $\mathbf{z}\in\mathbb{F}_q^n$ such that $\left|(\mathbf{z}+D)\cap B_q^n(\mathbf{0},p)\right|\geq t.$
\end{lemma}

\begin{proof}
Let $L:=B_q^n(\mathbf{y},p)\cap \cC$ and thus $L-\mathbf{y}=B_q^n(\mathbf{0},p)\cap (\cC-\mathbf{y})$. 
Assume that $L=\{\mathbf{x}_1,\ldots,\mathbf{x}_t\}$ with affine dimension $r$, then $L-\mathbf{y}\subseteq \operatorname{span}\{\mathbf{x}_2-\mathbf{x}_1,\ldots,\mathbf{x}_t-\mathbf{x}_1\}+\mathbf{x}_1-\mathbf{y}$. Thus, it suffices to choose $D=\operatorname{span}\{\mathbf{x}_2-\mathbf{x}_1,\ldots,\mathbf{x}_t-\mathbf{x}_1\}$ with dimension $r$ and $\mathbf{z}=\mathbf{x}_1-\mathbf{y}$. 
It is clear that $|L|=|L-\mathbf{y}|$ and thus $|(D+\mathbf{z})\cap B_q^n(\mathbf{0},p)|\geq t$. 

\end{proof}

By Lemma~\ref{lem:affine-reformulation}, a list of affine dimension $r$ implies that $\cC$ contains a $r$ dimensional subspace $D$. The following lemma bound the probability that $D\subseteq \cC$ for some fixed $D$.

\begin{lemma}
\label{lem:fixed-subspace}
Let $D\subseteq \mathbb{F}_q^n$ be a $r$ dimensional subspace and let $\cC\subseteq \mathbb{F}_q^n$ be a uniformly random $k$ dimensional subspace where $r\le k$. Then
\[
\Pr_{\cC}[D\subseteq \cC]
=
\prod_{j=0}^{r-1}
\frac{q^k-q^j}{q^n-q^j}.
\]
In particular, if $r=O_{p,q,\varepsilon}(1)$ and $n\to\infty$, then $\Pr[D\subseteq \cC]=q^{-(n-k)r+o(n)}.$
\end{lemma}

Suppose that an affine coset $\mathbf{z}+D$ contains at least $t$ points of $B_q^n(\mathbf{0},p)$. Choose any $t$ distinct such points, and let $r$ denote their affine rank. Choose an affine basis $\mathbf{x}_1,\ldots,\mathbf{x}_{r+1} \in B_q^n(\mathbf{0},p)$ from these $t$ points. Every one of the $t$ points can then be written uniquely as $\sum_{i=1}^{r+1} a_i\mathbf{x}_i$ for some $(a_1,\ldots,a_{r+1})\in\mathcal{A}_{r+1}$.
Thus, these $t$ points defines a set $A$ such that
\[
A\subseteq\mathcal{A}_{r+1},
\qquad
|A|=t,
\qquad
\mathbf{e}_1,\ldots,\mathbf{e}_{r+1}\in A.
\]
 Since $|\mathcal{A}_{r+1}|=q^r$, for fixed $r$ and $t$ the number of possible sets $A$ is at most $\binom{q^r}{\,t-r-1\,}\le q^{rt}=q^{O(t^2)}.$

For a fixed coefficient set $A$, define
\[
N(A)
:=
\left\{
(\mathbf{x}_1,\ldots,\mathbf{x}_{r+1})
\in
B_q^n(\mathbf{0},p)^{r+1}
:
\sum_{i=1}^{r+1}a_i\mathbf{x}_i
\in B_q^n(\mathbf{0},p)
\ \forall\
\mathbf{a}=(a_1,\ldots,a_{r+1})\in A
\right\}.
\]
For every affinely independent tuple $\mathbf{x}=(\mathbf{x}_1,\ldots,\mathbf{x}_{r+1})\in N(A)$, define the $r$ dimensional subspace $D_{\mathbf{x}}:=\mathrm{span}\left\{\mathbf{x}_2-\mathbf{x}_1,\ldots,\mathbf{x}_{r+1}-\mathbf{x}_1\right\}.$ Since $\mathbf{x}_1,\ldots,\mathbf{x}_{r+1}$ are affinely independent, we have $\dim D_{\mathbf{x}}=r.$ By Lemma~\ref{lem:affine-reformulation}, it suffices to bound the probability $D_{\mathbf{x}}\subseteq \cC$ and take a union bound over all tuples $(\mathbf{x}_1,\ldots,\mathbf{x}_{r+1})\in N(A)$ and $A\subseteq \F_q^{r+1}$.
Let $E_{t,r}$ be the event that there exists $\mathbf{y}\in\mathbb{F}_q^n$ such that $B_q^n(\mathbf{y},p)\cap \cC$ has size $t$ and affine dimension $r$. Then, we have
\[
\Pr_{\cC}\bigl[E_{t,r}\bigr]
\le
\sum_A
\sum_{\substack{
(\mathbf{x}_1,\ldots,\mathbf{x}_{r+1})\in N(A)\\
\text{affinely independent}
}}
\Pr\bigl[D_{\mathbf{x}}\subseteq \cC\bigr]
\le
\max_A |N(A)|\cdot q^{O(t^2)}
\cdot q^{-(h_q(p)+\varepsilon)rn+o(n)}.
\]
The last inequality follows from Lemma~\ref{lem:fixed-subspace}. Therefore, it suffices to upper-bound $|N(A)|$ for every $A$. To do this, we first bound the size of the intersection of two Hamming ball.

\begin{lemma}
\label{lem:shifted-ball-intersection}
Fix $\tau>0$. There exists a constant $c_\tau=c_\tau(p,q)>0$ such that, for every $\mathbf{z}\in\mathbb{F}_q^n$ satisfying $\mathrm{wt}(\mathbf{z})\ge \tau n,$ we have
\[
\left|
B_q^n(\mathbf{0},p)
\cap
B_q^n(\mathbf{z},p)
\right|
\le
q^{(h_q(p)-c_\tau)n}.
\]
In particular, one may take $c_\tau=\frac{2\tau^2}{\ln q}\left(1-\frac{qp}{q-1}\right)^2.$
\end{lemma}

\begin{proof}
Define a random vector $\mathbf{v}=(v_1,\ldots,v_n)\in\mathbb{F}_q^n$ with the following distribution
\[
\Pr[v_i=0]=1-p,
\qquad
\Pr[v_i=a]=\frac{p}{q-1}
\quad\text{for every }a\in \F_q^*.
\]
Since $p<1-1/q$, thus
$\Pr[\mathbf{v}=\mathbf{x}]\ge q^{-h_q(p)n}$ for every $\mathbf{x}\in B_q^n(\mathbf{0},p)$. Consequently, for every $S\subseteq B_q^n(\mathbf{0},p)$, we have $\Pr[\mathbf{v}\in S]=\sum_{\mathbf{x}\in S}\Pr[\mathbf{v}=\mathbf{x}]\ge |S|q^{-h_q(p)n}$ and hence $|S|\le q^{h_q(p)n}\Pr[\mathbf{v}\in S].$

Now fix a vector $\mathbf{z}\in\mathbb{F}_q^n$ with $\mathrm{wt}(\mathbf{z})\ge\tau n$. For each coordinate $i$, if $z_i=0$, then $\Pr[v_i-z_i\ne 0]=p.$ If $z_i\ne 0$, then $\Pr[v_i-z_i\ne 0]=1-\frac{p}{q-1}=p+\left(1-\frac{qp}{q-1}\right).$ It follows that $\mathbb{E}\bigl[\operatorname{wt}(\mathbf{v}-\mathbf{z})\bigr]
\ge
\left[
p+
\tau
\left(
1-\frac{qp}{q-1}
\right)
\right]n.$ Therefore, Hoeffding's inequality gives
\[
\Pr\left[
\mathrm{wt}(\mathbf{v}-\mathbf{z})\le pn
\right]
\le
\exp\left(
-2\tau^2
\left(
1-\frac{qp}{q-1}
\right)^2 n
\right).
\]
Then, we have
\begin{align*}
\left|
B_q^n(\mathbf{0},p)
\cap
B_q^n(\mathbf{z},p)
\right|
&\leq q^{h_q(p)n}\Pr[\mathbf{v}\in 
B_q^n(\mathbf{0},p)
\cap
B_q^n(\mathbf{z},p)] \\
&\le
q^{h_q(p)n}
\Pr\left[
\operatorname{wt}(\mathbf{v}-\mathbf{z})\le pn
\right]
\\
&\le
q^{h_q(p)n}
\exp\left(
-2\tau^2
\left(
1-\frac{qp}{q-1}
\right)^2n
\right)
\\
&=
q^{(h_q(p)-c_\tau)n},
\end{align*}
where $c_\tau=\frac{2\tau^2}{\ln q}\left(1-\frac{qp}{q-1}\right)^2.$
\end{proof}

We next use a simple variant of the increasing chain argument in~\cite[Lemma~12]{guruswami2010list}. In fact, we only need a weaker result which simplifies our argument.

\begin{lemma}\label{lem:anchored-chain}
Let $S \subseteq \mathbb{F}_q^m$ with $|S|=L$. Then there exist $\mathbf{a}_0,\mathbf{a}_1,\ldots,\mathbf{a}_h\in S$ with $h \ge \lfloor \log_q L \rfloor$ such that $\mathrm{supp}(\mathbf{a}_j-\mathbf{a}_0) \not\subseteq\bigcup_{i<j}\mathrm{supp}(\mathbf{a}_i-\mathbf{a}_0)$ for $1\leq j\leq h$.
\end{lemma}

\begin{proof}
We prove the lemma by induction on $L$. If $L<q$, then $\lfloor \log_q L \rfloor=0$, so there is nothing to prove.

Now suppose $L\ge q$. Since $S$ contains at least two distinct vectors, there is a coordinate $j$ on which not all vectors in $S$ agree. For each $\alpha\in\mathbb{F}_q$, let $S_\alpha:=\{\mathbf{x}\in S:x_j=\alpha\}.$ Choose $\alpha$ for which $|S_\alpha|$ is largest and write it as $S_0$. Since there are at most $q$ possible values of the $j$-th coordinate, $\frac{L}{q}\leq |S_0|<L$. By the induction hypothesis, there exist $\mathbf{a}_0,\mathbf{a}_1,\ldots,\mathbf{a}_h\in S_0$, where $h\ge \left\lfloor \log_q |S_0| \right\rfloor$ such that $\mathrm{supp}(\mathbf{a}_i-\mathbf{a}_0)\not\subseteq\cup_{\ell<i}\mathrm{supp}(\mathbf{a}_\ell-\mathbf{a}_0)$ for $i=1,\ldots,h$.

Choose any $\mathbf{a}_{h+1}\in S\setminus S_0$. All vectors in $S_0$ have the same $j$-th coordinate, so $(\mathbf{a}_i-\mathbf{a}_0)_j=0$ whereas $(\mathbf{a}_{h+1}-\mathbf{a}_0)_j\ne 0.$ Thus, $\mathrm{supp}(\mathbf{a}_{h+1}-\mathbf{a}_0)\not\subseteq\cup_{\ell\leq h}\mathrm{supp}(\mathbf{a}_\ell-\mathbf{a}_0)$. 
Finally, $h+1\ge 1+\left\lfloor \log_q |S_0| \right\rfloor\ge 1+\left\lfloor \log_q \frac{L}{q} \right\rfloor = \left\lfloor \log_q L \right\rfloor .$
\end{proof}

We now proceed to bounding $|N(A)|$. A trivial bound is $|N(A)| \le \bigl|B_q^n(\mathbf{0},p)\bigr|^{r+1},$ by counting the number of tuples $(\mathbf{x}_1,\ldots,\mathbf{x}_{r+1})$ in $B_q^n(\mathbf{0},p)^{r+1}$. This bound does not yield what we want. Thus, we need to exploit the fact that 
for every coefficient vector $\mathbf{a}=(a_1,\ldots,a_{r+1})\in A\setminus\{\mathbf{e}_1,\ldots,\mathbf{e}_{r+1}\}$, it holds that $\sum_{i=1}^{r+1} a_i\mathbf{x}_i \in B_q^n(\mathbf{0},p).$ We next use it to obtain a sharper bound on $|N(A)|$.

Let $S:=A\setminus\{\mathbf{e}_1,\ldots,\mathbf{e}_{r+1}\}.$ Then, the support set of every $\mathbf{a}\in S$ has size at least $2$. Indeed, if $\mathrm{supp}(\mathbf{a})=\{i\}$, then $\sum_{j=1}^{r+1} a_j=1$ implies $\mathbf{a}=\mathbf{e}_i$, contradicting $\mathbf{a}\in S$. Applying Lemma~\ref{lem:anchored-chain} to $S$, we can choose $\mathbf{a}_0,\mathbf{a}_1,\ldots,\mathbf{a}_h\in S$ with $h\ge \lfloor \log_q |S| \rfloor$ and set $\mathbf{b}_j:=\mathbf{a}_j-\mathbf{a}_0$ for $1\leq j\leq h$, so that each $\mathbf{b}_j$ with $(\mathbf{b}_j)_{i_j}\neq 0$ and $(\mathbf{b}_\ell)_{i_j}=0$ for every $\ell<j$.

\begin{lemma}\label{lem:triangularization}   
Using the above notations and let $X:=\{\sum_{\ell=1}^{r+1}(\mathbf{a}_j)_\ell\mathbf{x}_\ell: j=0,\ldots,h\}$. Then, there exists a new tuple obtained by an invertible linear change of variables $\mathbf{x}'_1,\ldots,\mathbf{x}'_{r+1}\in \{\mathbf{x}_1,\ldots,\mathbf{x}_{r+1}, \sum_{\ell=1}^{r+1}(\mathbf{a}_0)_\ell\mathbf{x}_\ell\}$ and a set $J\subseteq[h]$ with $|J|\ge h-1$ such that, for every $j\in J$, we have $\sum_{\ell=1}^{r+1}(a_j)_\ell \mathbf{x}_\ell=\beta_j \mathbf{x}'_{i_j}+\mathbf{z}_j\in B_q^n(0,p)$ where $\beta_j\in\mathbb{F}_q^*$ and $\mathbf{z}_j$ is a linear combination of the vectors $\mathbf{x}'_i$ with $i\notin\{i_k:k\in J,\ k\ge j\}.$
\end{lemma}

\begin{proof}
Recall that $\mathbf{b}_j:=\mathbf{a}_j-\mathbf{a}_0$ and that the coordinates $i_1,\ldots,i_h$ were chosen so that $(\mathbf{b}_j)_{i_j}\neq 0$ and $(\mathbf{b}_k)_{i_j}=0$ for every $k<j$.

We first consider the case $(\mathbf{a}_0)_{i_j}=0$ for every $1\leq j\leq h$. Set $J:=[h]$ and $\mathbf{x}'_i:=\mathbf{x}_i$ for every $i\in[r+1]$. Then, we have 
$$
\sum_{\ell=1}^{r+1}(\mathbf{a}_j)_\ell\mathbf{x}_\ell=
\sum_{\ell=1}^{r+1}(\mathbf{a}_j-\mathbf{a}_0)_\ell\mathbf{x}_\ell+\sum_{\ell=1}^{r+1}(\mathbf{a}_0)_\ell\mathbf{x}_\ell.
$$
It is clear that $i_{j+1},\ldots,i_h\notin \mathrm{supp}(\mathbf{a}_j-\mathbf{a}_0)\cup \mathrm{supp}(\mathbf{a}_0)$. Thus, we can choose $\mathbf{z}_j=\sum_{\ell=1}^{r+1}(\mathbf{a}_j)_\ell\mathbf{x}_\ell-(\mathbf{a}_j)_{i_j}\mathbf{x}_{i_j}$ to complete the construction.  


Now suppose that $(\mathbf{a}_0)_{i_j}\neq 0$ for some $j$. Let $s:=\max\{j:(\mathbf{a}_0)_{i_j}\neq 0\}$ and set $\mathbf{y}:=\sum_{\ell=1}^{r+1}(\mathbf{a}_0)_\ell\mathbf{x}_\ell.$ Set $J:=[h]\setminus\{s\}$. We replace $\mathbf{x}_{i_s}$ with $\mathbf{y}$ according to the relation $\mathbf{x}_{i_s}=\frac{1}{(\mathbf{a}_0)_{i_s}}\left(\mathbf{y}-\sum_{\ell\neq i_s}(\mathbf{a}_0)_\ell\mathbf{x}_\ell\right)$. It is clear that these $r+1$ vectors $\mathbf{y},\, \{\mathbf{x}_\ell:\ell\neq i_s\}$ all lying in $B_q^n(\mathbf{0},p)$. Let $\mathbf{x}'_i=\mathbf{x}_i$ for $i\neq i_s$ and $\mathbf{x}'_{i_s}=\mathbf{y}$.

Assume first $j<s$ and we have $(\mathbf{b}_j)_{i_s}=0.$ Therefore 
$$\sum_{\ell=1}^{r+1}(\mathbf{a}_j)_\ell\mathbf{x}_\ell=\mathbf{y}+\sum_{\ell=1}^{r+1}(\mathbf{b}_j)_\ell\mathbf{x}_\ell=\beta_j\mathbf{x}_{i_j}+\mathbf{z}_j$$ 
where $\beta_j:=(\mathbf{b}_j)_{i_j}\neq 0$ and $\mathbf{z}_j:=\mathbf{y}+\sum_{k<j}(\mathbf{b}_j)_{i_k} \mathbf{x}_{i_k}+\sum_{\ell\notin\{i_1,\ldots,i_h\}}(\mathbf{b}_j)_\ell\mathbf{x}_\ell$. It is clear that $\mathbf{z}_j$ meets the requirement.

Next we assume that $j>s$. In this case, we write
\[
\begin{aligned}
\sum_{\ell=1}^{r+1}(\mathbf{a}_j)_\ell\mathbf{x}_\ell
={}&
\left(
1+\frac{(\mathbf{b}_j)_{i_s}}{(\mathbf{a}_0)_{i_s}}
\right)\mathbf{y} +
\sum_{\ell\neq i_s}
\left(
(\mathbf{b}_j)_\ell
-
\frac{(\mathbf{b}_j)_{i_s}}{(\mathbf{a}_0)_{i_s}}
(\mathbf{a}_0)_\ell
\right)\mathbf{x}_\ell.
\end{aligned}
\]
Since $j>s$, the definition of $s$ implies $(\mathbf{a}_0)_{i_j}=0.$ Hence the coefficient of $\mathbf{x}_{i_j}$ is $(\mathbf{b}_j)_{i_j}\neq 0.$ Moreover, for every $k>j$, we have $(\mathbf{b}_j)_{i_k}=0$ and $(\mathbf{a}_0)_{i_k}=0,$ the coefficient of $\mathbf{x}_{i_k}$ is $0$. Thus, we write $\sum_{\ell=1}^{r+1}(\mathbf{a}_j)_\ell\mathbf{x}_\ell=\beta_j\mathbf{x}_{i_j}+\mathbf{z}_j$ where $\beta_j:=(\mathbf{b}_j)_{i_j}\neq 0$ and
\[
\mathbf{z}_j
={}
\left(
1+\frac{(\mathbf{b}_j)_{i_s}}{(\mathbf{a}_0)_{i_s}}
\right)\mathbf{y} +\sum_{\substack{k<j\\k\ne s}}
\left(
(\mathbf{b}_j)_{i_k}
-
\frac{(\mathbf{b}_j)_{i_s}}{(\mathbf{a}_0)_{i_s}}
(\mathbf{a}_0)_{i_k}
\right)\mathbf{x}_{i_k} +
\sum_{\ell\notin\{i_1,\ldots,i_h\}}
\left(
(\mathbf{b}_j)_\ell
-
\frac{(\mathbf{b}_j)_{i_s}}{(\mathbf{a}_0)_{i_s}}
(\mathbf{a}_0)_\ell
\right)\mathbf{x}_\ell .
\]
We skip the case $j=s$ and thus yields $h-1$ elements of $X$ meeting the requirement.

Finally, the linear form $\mathbf{z}_j$ has at least one nonzero coefficient. Indeed, the change of variables is invertible and $\mathbf{x}'_{i_j}=\mathbf{x}_{i_j}$ for every $j\in J$. If all coefficients of $\mathbf{z}_j$ were zero, then $\mathbf{a}_j=\beta_j\mathbf{e}_{i_j}$. Note that $\mathbf{a}_j$ belongs to the affine subspace $\mathcal{A}_{r+1}$ which forces that $\beta_j=1$. This is absurd as $\mathbf{a}_j\in A\setminus\{\mathbf{e}_1,\ldots,\mathbf{e}_{r+1}\}.$ The proof is completed.

\end{proof}

We now use Lemma~\ref{lem:triangularization}, together with the intersection of two Hamming balls to derive a sharper upper bound on $|N(A)|$.

\begin{lemma}
\label{lem:affine-saturation}
Fix $p,q$. For every $\eta>0$, there exists a constant $M=M(p,q,\eta)$ such that the following holds. Suppose
\[
A\subseteq \mathcal{A}_{r+1},
\qquad
\mathbf{e}_1,\ldots,\mathbf{e}_{r+1}\in A,
\qquad
|A|-(r+1)\ge M.
\]
Then, for all sufficiently large $n$, $|N(A)|\le q^{[rh_q(p)+\eta]n}.$
\end{lemma}

\begin{proof}
Choose $0<\tau<p$ sufficiently small so that $h_q(\tau)<\frac{\eta}{4}.$ By Lemma~\ref{lem:shifted-ball-intersection}, there exists a constant $c_\tau>0$ such that $|B_q^n(\mathbf{0},p)\cap B_q^n(\mathbf{z},p)|\le q^{(h_q(p)-c_\tau)n}$ whenever $\mathrm{wt}(\mathbf{z})\ge \tau n$. Choose an integer $m$ such that $mc_\tau\ge h_q(p)$ and set $M:=q^{m+1}.$

Recall that $S:=A\setminus\{\mathbf{e}_1,\ldots,\mathbf{e}_{r+1}\}.$ Since $|S|=|A|-(r+1)\ge M,$ Lemma~\ref{lem:anchored-chain} gives $\mathbf{a}_0,\mathbf{a}_1,\ldots,\mathbf{a}_h\in S$ and $h\ge \lfloor \log_q |S| \rfloor\ge m+1$ such that, writing $\mathbf{b}_j:=\mathbf{a}_j-\mathbf{a}_0,$ for each $1\le j\le h$ there is a coordinate $i_j$ satisfying $(\mathbf{b}_j)_{i_j}\neq 0$ and $(\mathbf{b}_j)_{i_k}=0$ for every $k>j$. By Lemma~\ref{lem:triangularization}, we can find a basis $\mathbf{x}'_1,\ldots,\mathbf{x}'_{r+1}\in B_q^n(0,p)$ and a set $J\subseteq[h]$ with $|J|\ge h-1\ge m$ such that, for every $j\in J$, we have $\beta_j\mathbf{x}'_{i_j}+\mathbf{z}_j\in B_q^n(0,p)$ where $\beta_j\in\mathbb{F}_q^*$ and $\mathbf{z}_j$ is a linear combination of $\mathbf{x}'_i$ with $i\notin\{i_k:k\in J,\ k\ge j\}$. Choose $m$ indices from $J$ in increasing order and relabel the corresponding basis vectors as $\mathbf{x}'_1,\ldots,\mathbf{x}'_m$. Relabeling the corresponding $\beta_j$ and $z_j$ as well, we may write these expressions as $\beta_j\mathbf{x}'_j+\mathbf{z}_j\in B_q^n(0,p)$ for $1\leq j\leq m$ where $\mathbf{z}_j$ is a linear combination of $\mathbf{x}'_1,\ldots,\mathbf{x}'_{j-1},\mathbf{x}'_{m+1},\ldots,\mathbf{x}'_{r+1}$.
Note that $\mathbf{x}_i, i\in [r+1]$ uniquely determine $\mathbf{x}'_i, i\in [r+1]$ and vice versa by Lemma \ref{lem:triangularization}.

We divide the tuples in $N(A)$ into two cases.

\medskip
\noindent
\textbf{Case 1: $\mathrm{wt}(\mathbf{z}_j)\le \tau n$ for some $1\le j\le m$.}

Fix an index $j$ and write $\mathbf{z}_j=\sum_{\ell=1}^{r+1}\gamma_\ell\mathbf{x}'_\ell.$ Since the coefficients of $\mathbf{z}_j$ are not all zero, choose an index $\ell$ with $\gamma_\ell\ne0$. Since $\mathrm{wt}(\mathbf{z}_j)\le \tau n,$ we may first choose $\mathbf{z}_j\in B_q^n(\mathbf{0},\tau)$ and then choose the other $r$ variables $\mathbf{x}'_k\,(k\neq \ell)$ from $B_q^n(\mathbf{0},p)$. After that, we have $\mathbf{x}'_\ell=\gamma_\ell^{-1}\left(\mathbf{z}_j-\sum_{k\neq \ell}\gamma_k\mathbf{x}'_k\right).$ Therefore, the number of such tuples is at most
\[
\left|B_q^n(\mathbf{0},\tau)\right|
\left|B_q^n(\mathbf{0},p)\right|^r
\le
q^{[rh_q(p)+h_q(\tau)]n}.
\]
Since $j$ range from $1$ to $m$, the total number of tuples is at most $m q^{[rh_q(p)+h_q(\tau)]n}.$ Since $m$ is independent of $n$ and $h_q(\tau)<\frac{\eta}{4}$, for all sufficiently large $n$, we obtain $m q^{[rh_q(p)+h_q(\tau)]n}\le q^{[rh_q(p)+\eta/2]n}.$

\medskip
\noindent
\textbf{Case 2: $\mathrm{wt}(\mathbf{z}_j)>\tau n$ for every $1\le j\le m$.}

First fix all vectors in $\mathbf{x}'_1,\ldots,\mathbf{x}'_{r+1}$ other than $\mathbf{x}'_1,\ldots,\mathbf{x}'_m.$ There are at most $|B_q^n(\mathbf{0},p)|^{r+1-m}$ ways to do this.

We then consider $\mathbf{x}'_1,\mathbf{x}'_2,\ldots,\mathbf{x}'_m$ in this order. We start from $j=1$ to $j=m$. By the above relabeling, $\mathbf{z}_j$ is determined by $\mathbf{x}'_1,\ldots,\mathbf{x}'_{j-1},\mathbf{x}'_{m+1},\ldots,\mathbf{x}'_{r+1}$, so it is fixed when we count the choices for $\mathbf{x}'_j$. Since $\mathbf{x}'_j$ satisfies the requirement that $\mathbf{x}'_j\in B_q^n(\mathbf{0},p)$ and $\beta_j\mathbf{x}'_j+\mathbf{z}_j\in B_q^n(\mathbf{0},p)$, Lemma~\ref{lem:shifted-ball-intersection} therefore shows that there are at most $q^{(h_q(p)-c_\tau)n}$ possible choices for $\mathbf{x}'_j$.

Applying this bound successively for $j=1,\ldots,m$, the number of tuples in this case is at most $|B_q^n(\mathbf{0},p)|^{r+1-m}q^{m(h_q(p)-c_\tau)n}.$ Using $|B_q^n(\mathbf{0},p)|\le q^{h_q(p)n},$ we obtain
\[
\begin{aligned}
|B_q^n(\mathbf{0},p)|^{r+1-m}
q^{m(h_q(p)-c_\tau)n}
\le
q^{[(r+1)h_q(p)-mc_\tau]n}\le
q^{rh_q(p)n},
\end{aligned}
\]
where the last inequality follows from $mc_\tau\ge h_q(p).$

Combining the two cases, for all sufficiently large $n$, $|N(A)|\le q^{[rh_q(p)+\eta]n}.$

\end{proof}

The lemma above shows that, once the size of $A$ is large enough, the number of tuples in $N(A)$ can be bounded by $|N(A)|\le q^{[rh_q(p)+\eta]n}.$ To complete the proof of the main theorem, we will also use the following result in~\cite{guruswami2010list}.

\begin{thm}[{\cite[Theorem~10]{guruswami2010list}}]
\label{thm:GHK}
For every prime power $q$ and every $p\in(0,1-1/q)$, there is a constant $C_0>1$, such that for all $n$ large enough and all $\ell=o(\sqrt{n})$, if $\mathbf{X}_1,\ldots,\mathbf{X}_\ell$ are picked independently and uniformly at random from $B_q^n(\mathbf{0},p)$, then
\[
\Pr\left[
\left|
\mathrm{span}(\{\mathbf{X}_1,\ldots,\mathbf{X}_\ell\})
\cap B_q^n(\mathbf{0},p)
\right|
>
C_0\cdot\ell
\right]
\le q^{-5n}.
\]
\end{thm}

We now prove the main theorem.

\begin{thm}
\label{thm:main}
Fix a prime power $q$ and $0<p<1-1/q$. For every fixed $\varepsilon>0$ with $1-h_q(p)-\varepsilon>0,$ let $\cC\subseteq\mathbb{F}_q^n$ be a uniformly random linear code of rate $R=1-h_q(p)-\varepsilon.$ Then, with probability at least $1-q^{-\Omega_{p,q}(\varepsilon n)},$ the code $\cC$ is $\left(p,\,\left\lceil\frac{h_q(p)}{\varepsilon}\right\rceil+O_{p,q}(1) \right)$-list-decodable.
\end{thm}

\begin{proof}
Let $C_0$ be the constant in Theorem~\ref{thm:GHK}. Choose $\eta:=\frac{h_q(p)}{4C_0}$ and let $M=M(p,q,\eta)$ be the constant given by Lemma~\ref{lem:affine-saturation}. Finally, choose an integer $c>M+C_0+2$ and set $t:=\left\lceil\frac{h_q(p)}{\varepsilon}\right\rceil+c.$ We show that, with probability $1-q^{-\Omega_{p,q}(\varepsilon n)}$, no Hamming ball of radius $pn$ contains $t$ codewords of $\cC$.

Suppose that there exists a list $L$ of size $t$ with affine dimension $r$. Let
 $\mathbf{x}_1,\ldots,\mathbf{x}_{r+1}$ be the affine basis of dimension $r$. $L$ together with this affine basis determines the set $A$ such that
\[
A\subseteq\mathcal{A}_{r+1},
\qquad
|A|=t,
\qquad
\mathbf{e}_1,\ldots,\mathbf{e}_{r+1}\in A.
\]
Let $D_{\mathbf{x}}:=\mathrm{span}\{\mathbf{x}_2-\mathbf{x}_1,\ldots, \mathbf{x}_{r+1}-\mathbf{x}_1\}$ and by Lemma \ref{lem:affine-reformulation}, $D_{\mathbf{x}}\subseteq \cC$. By Lemma~\ref{lem:fixed-subspace}, $\Pr[D_{\mathbf{x}}\subseteq \cC]=q^{-(h_q(p)+\varepsilon)rn+o(n)}.$

We divide our discussion into two cases.

\medskip
\noindent
\textbf{Case 1: $r+1<\frac{t}{C_0}.$}
Let $\mathbf{x}_1,\ldots,\mathbf{x}_{r+1}$ be an affine basis defined above. Every vector in $L$ can be written as $\sum_{i=1}^{r+1} a_i\mathbf{x}_i$ where $\sum_{i=1}^{r+1} a_i=1.$ In particular, all $t$ vectors belong to $\mathrm{span}\{\mathbf{x}_1,\ldots,\mathbf{x}_{r+1}\}.$ Hence, $\left|\mathrm{span}\{\mathbf{x}_1,\ldots,\mathbf{x}_{r+1}\}\cap B_q^n(\mathbf{0},p)\right|\ge t>C_0(r+1)$. Since $r+1\le t=O_{p,q,\varepsilon}(1)=o(\sqrt n),$ Theorem~\ref{thm:GHK} shows that the number of tuples $(\mathbf{x}_1,\ldots,\mathbf{x}_{r+1})\in B_q^n(\mathbf{0},p)^{r+1}$ satisfying this condition is at most $\left|B_q^n(\mathbf{0},p)\right|^{r+1}q^{-5n}.$

Since $\Pr[D_{\mathbf{x}}\subseteq \cC]=q^{-(h_q(p)+\varepsilon)rn+o(n)}$, the probability of a bad list of size $t$ and affine dimension $r$ is at most
\[
\begin{aligned}
\left|B_q^n(\mathbf{0},p)\right|^{r+1}\cdot q^{-5n}\cdot q^{-(h_q(p)+\varepsilon)rn+o(n)}\le q^{[h_q(p)(r+1)-5-(h_q(p)+\varepsilon)r]n+o(n)}=q^{-\Omega(n)}.
\end{aligned}
\]

\medskip
\noindent
\textbf{Case 2: $r+1\geq \frac{t}{C_0}.$}

 The list $L$ yields $A$ and an affinely independent tuple $\mathbf{x}=(\mathbf{x}_1,\ldots,\mathbf{x}_{r+1})\in N(A)$. Moreover, $D_{\mathbf{x}}$ must satisfy $D_{\mathbf{x}}\subseteq \cC$. Let $E_{t,r}$ be the event that $L$ is a list of size $t$ and affine dimension $r$. Since there are at most $q^{O(t^2)}$ different sets $A$, and for each fixed affinely independent tuple $\Pr[D_{\mathbf{x}}\subseteq \cC]=q^{-(h_q(p)+\varepsilon)rn+o(n)},$ a union bound gives
\[
\Pr[E_{t,r}]
\le \max_A |N(A)|\cdot 
q^{O(t^2)}\cdot q^{-(h_q(p)+\varepsilon)rn+o(n)}.
\]
We now bound $|N(A)|$ according to whether $t-(r+1)\ge M$ or $t-(r+1)<M$.

If $t-(r+1)\ge M,$ Lemma~\ref{lem:affine-saturation} gives $|N(A)|\le q^{[rh_q(p)+\eta]n}.$ Therefore,
\[
\Pr[E_{t,r}]
\le
q^{O(t^2)}
q^{(\eta-\varepsilon r)n+o(n)}.
\]
Since $t\le C_0(r+1)$ and $t=\left\lceil\frac{h_q(p)}{\varepsilon}\right\rceil+c$, $c\ge C_0$, this gives $\varepsilon r\ge \frac{h_q(p)}{C_0}.$ Hence, $\eta-\varepsilon r\le \frac{h_q(p)}{4C_0}-\frac{h_q(p)}{C_0}<0.$ Thus, the probability in this case is at most $q^{-\Omega_{p,q}(n)}.$

If $t-(r+1)<M,$ we use the trivial bound $|N(A)|\le \left|B_q^n(\mathbf{0},p)\right|^{r+1}\le q^{h_q(p)(r+1)n}.$ Therefore,
\[
\Pr[E_{t,r}]\le q^{O(t^2)} q^{(h_q(p)-\varepsilon r)n+o(n)}.
\]
Since $r\ge t-M$, $h_q(p)-\varepsilon r\le h_q(p)-\varepsilon(t-M)\leq-\varepsilon(c-M)<-\varepsilon C_0$. Hence, the probability in this case is $q^{-\Omega_{p,q}(\varepsilon n)}.$

Finally, the affine dimension satisfies $1\le r\le t-1$. Since $t$ is independent of $n$ for fixed $\varepsilon$, we obtain $\sum_{r=1}^{t-1}\Pr[E_{t,r}]\le q^{-\Omega_{p,q}(\varepsilon n)}.$ The proof is completed.
\end{proof}

\begin{rmk}
    Fix any sufficiently small constant $\tau\in (0,p)$ such that $h_q(\tau)<\frac{h_q(p)}{16C_0}$, then the constant $M$ can be taken as $M=q^{\left\lceil\frac{h_q(p)\ln q}{2\tau^2\left(1-\frac{qp}{q-1}\right)^2}\right\rceil+1}$ which is clearly independent of $\epsilon$. 
\end{rmk}

\section{Lower Bounds for Average Radius List Decoding}\label{sec:average}
Throughout this section, fix a prime power $q$ and $0<p<1-\frac{1}{q}.$ For convenience, write $h_q'(p)=\log_q\frac{(q-1)(1-p)}{p}$ which is the derivative of $h_q(x)=x\log_q(q-1)-x\log_q x-(1-x)\log_q (1-x)$ at $x=p$.
Define $K:=p(1+h'_q(p)).$ For $\varepsilon>0$, we consider a random linear code $\cC\subseteq\mathbb{F}_q^n$ of rate $R:=1-h_q(p)-\varepsilon.$ The following is our main theorem in this section.

\begin{thm}
\label{thm:average radius-lower bound}
Fix a prime power $q$, $p\in(0,1-1/q)$ and $\delta\in(0,p(1+h'_q(p)))$. There exists $\varepsilon_{p,q,\delta}>0$ such that, for every $\varepsilon\in(0,\varepsilon_{p,q,\delta})$ and all sufficiently large $n$, a random linear code $\cC\subseteq \mathbb{F}_q^n$ of rate $1-h_q(p)-\varepsilon$ is not
\[
\left(
p,
\left\lfloor
\frac{p(1+h'_q(p))-\delta}{\varepsilon}
\right\rfloor
\right)
\text{ average radius list decodable}
\]
with probability $1-q^{-\Omega_{p,q,\delta,\varepsilon}(n)}.$
\end{thm}

We call a distribution $\tau$ on $\mathbb{F}_q^L$ \emph{bad} for $(p,L)$ average radius list decoding if the following holds: for every $M\in\mathcal{M}_{n,\tau}$, writing $M_1,\ldots,M_L$ for the columns of $M$ and they are pairwise distinct, there exists some $\mathbf{z}\in\mathbb{F}_q^n$ such that $\frac{1}{L}\sum_{j=1}^L d(M_j,\mathbf{z})\le pn.$ Thus, if $M\subseteq \cC$ for $M\in\mathcal{M}_{n,\tau}$, then the columns of $M$ form a bad list in $\cC$ and hence $\cC$ is not $(p,L)$ average radius list decodable.

We say that $\tau$ is \emph{abundant} for rate $R$ if a random linear code $\cC\subseteq\mathbb{F}_q^n$ of rate $R$ contains some matrix $M\in\mathcal{M}_{n,\tau}$ with high probability, i.e.,
\[
\Pr\!\left[
    \exists\,M\in\mathcal{M}_{n,\tau}: M\subseteq \cC
\right]
=1-o(1)
\]
as $n\to\infty$. Clearly, if $\tau$ is both bad and abundant for rate $R$, then a random linear code of rate $R$ is not $(p,L)$ average radius list decodable with high probability.

Our proof of Theorem~\ref{thm:average radius-lower bound} follows the same general outline as the lower bound argument in~\cite[Section~4]{GuruswamiEtAl2022}. We first define a bad distribution $\tau$ in Definition~\ref{def:bad-distribution}; we then show that $\tau$ is indeed bad in Proposition~\ref{prop:tau-bad}. Next, we prove that $\tau$ is not $\gamma$-implicitly rare in Lemma~\ref{lem:not-implicitly-rare}, and hence is abundant by Theorem~\ref{thm:abundance}. Finally, we combine these results to prove Theorem~\ref{thm:average radius-lower bound}.

Fix a sufficiently large integer $m$ and set $L:=m+q.$ Let $\eta_{q,p,m}=\frac{p(1+h'_q(p))+2}{h'_q(p)m},$ and
define
\[
p_1
:=
p+\frac{p}{m}
-\frac{(q-1)\eta_{q,p,m}}
       {m},
\qquad
p_2
:=
p+\eta_{q,p,m}.
\]
It is clear that $p_1,p_2\in (0,1-1/q)$. We use the shorthand $\eta_{q,p,m}=\eta$ if no ambiguity arises.
\begin{defn}[The bad distribution $\tau$]
\label{def:bad-distribution}
Let
\[
\alpha\sim\mathrm{Unif}(\mathbb{F}_q),\qquad
E_1,\ldots,E_m
\stackrel{\mathrm{i.i.d.}}{\sim}
\mathrm{Bernoulli}_q(p_1),
\qquad
T\sim\mathrm{Bernoulli}_q(p_2),
\]
where $\alpha,E_1,\ldots,E_m,T$ are mutually independent. Conditioned on $T$, define a random variable $B$ as follows: if $T=0$, let $B=0$; if $T\neq 0$, let $B$ be uniformly distributed over $\mathbb{F}_q$, independently of $\alpha,E_1,\ldots,E_m$. Define $\tau$ on $\mathbb{F}_q^L$ as the distribution of the random vector
\[
\mathbf{X}=\left(\alpha+E_1,\ldots,\alpha+E_m,\,(\alpha+B+aT)_{a\in\mathbb{F}_q}\right).
\]
\end{defn}

We show that $\tau$ is bad for average radius list decoding.

\begin{prop}
\label{prop:tau-bad}
Let $\tau$ be as in Definition~\ref{def:bad-distribution}. Let $\cC\subseteq\mathbb{F}_q^n$ and $M\in\mathcal{M}_{n,\tau}$. If $M\subseteq \cC$, $\cC$ is not $(p,L)$ average radius list decodable for all sufficiently large $n$.
\end{prop}

\begin{proof}
Let $M\in\mathcal{M}_{n,\tau}$ and $M_1,\ldots,M_L$ denote its columns. We want to show that there exists a vector $\mathbf{z}\in\mathbb{F}_q^n$ such that $\frac{1}{L}\sum_{j=1}^L d(M_j,\mathbf{z})\le pn$.

By definition of $\tau$ and $\mathcal{M}_{n,\tau}$, we write the $r$-th row of $M$ as
\[
M_{r,*}
=
\left(
\alpha_r+E_{r,1},\ldots,\alpha_r+E_{r,m},
(\alpha_r+B_r+aT_r)_{a\in\mathbb{F}_q}
\right).
\]

Set $\mathbf{z}:=(\alpha_1,\ldots,\alpha_n)\in\mathbb{F}_q^n$. For each of the first $m$ columns which are generated by the random variables $\alpha+E_i$, the element in row $r$ is not equal to $\alpha_r$ when $E_{r,i}\neq 0$. For the rest $q$ columns which are generated by random variable $\alpha+B+aT$, if $T_r=0$, then all of them are equal to $\alpha_r$; if $T_r\neq 0$, then $a\mapsto B_r+aT_r$ is a permutation of $\mathbb{F}_q$ and hence exactly $q-1$ of them differ from $\alpha_r$. Therefore,
\[
\begin{aligned}
\frac{1}{L}\sum_{j=1}^L d(M_j,\mathbf{z})
&=
\frac{1}{L}\sum_{r=1}^n
\left(
\sum_{i=1}^m \mathbf{1}[E_{r,i}\neq 0]
+
(q-1)\mathbf{1}[T_r\neq 0]
\right) \\
&\le
\frac{n}{L}\bigl(mp_1+(q-1)p_2\bigr)
=np.
\end{aligned}
\]
For each row, the chosen representation has no more disagreements than the average under Definition~\ref{def:bad-distribution}, conditioned on producing that row. Since $M$ has type $\tau$, averaging over its rows gives the inequality.

Finally, for any two coordinates $i\neq j$, there exists a vector in $\mathrm{supp}(\tau)$ whose $i$-th and $j$-th elements are different. Thus, for all sufficiently large $n$, the columns $M_1,\ldots,M_L$ are pairwise distinct. Therefore, they form a bad list of size $L$ in $\cC$, and $\cC$ is not $(p,L)$ average radius list decodable.
\end{proof}

Next, we show that $\tau$ is not implicitly rare, which will imply that $\tau$ is abundant by Theorem~\ref{thm:abundance}.

\begin{lemma}
\label{lem:not-implicitly-rare}
Let $p\in(0,1-1/q)$ and let $0<\zeta<\min\left\{\frac{p}{4},\frac{p(1+h'_q(p))}{4}\right\}.$ There exists $m_{p,q,\zeta}>0$ such that, for every $m\ge m_{p,q,\zeta}$, the distribution $\tau$ given in Definition~\ref{def:bad-distribution} is not $\gamma$-implicitly rare where $\gamma:=h_q(p)+\frac{p(1+h'_q(p))-\zeta}{m}.$
\end{lemma}
Note that 
$\gamma=h_q(p)+\frac{p}{m}h'_q(p)+\frac{p-\zeta}{m}$. Since $p_1=p+\frac{p}{m}-\frac{(q-1)\eta}{m}$, we have $h_q(p_1)=h_q(p)+(\frac{p-(q-1)\eta}{m})h'_q(p)+O(\frac{1}{m^2})$ by Taylor expansion. Thus, for large enough $m$, we have
\begin{equation}\label{eq:Taylor}
\gamma-h_q(p_1)\leq (p-\zeta)/m+O_{p,q}(m^{-2})\leq (p-\zeta/2)/m.
\end{equation}
Moreover, since $p_2=p+\frac{p}{m}+\frac{p+2}{mh'_q(p)}$, this implies 
\begin{equation}\label{eq:gamma}
h_q(p_2)=h_q(p)+\frac{p}{m}h'_q(p)+\frac{p+2}{m}+O(\frac{1}{m^2})\geq \gamma
\end{equation}
for large enough $m$.
We postpone the proof of Lemma~\ref{lem:not-implicitly-rare} until later. We first use it to prove Theorem~\ref{thm:average radius-lower bound}. 

\begin{proof}[Proof of Theorem~\ref{thm:average radius-lower bound}]
Fix $\delta>0$. Choose a constant $\zeta$ such that $0<\zeta<\min\left\{\frac{\delta}{2},\frac{p}{4},\frac{p(1+h'_q(p))}{4}\right\}.$ Set $m:=\left\lfloor\frac{p(1+h'_q(p))-2\zeta}{\varepsilon}\right\rfloor$ and recall that $L=m+q.$

For sufficiently small $\varepsilon$, by Lemma~\ref{lem:not-implicitly-rare}, the distribution $\tau$ in Definition~\ref{def:bad-distribution} is not $\gamma$-implicitly rare where $\gamma=h_q(p)+\frac{p(1+h'_q(p))-\zeta}{m}.$ Thus, we have $\gamma-(1-R)=\frac{p(1+h'_q(p))-\zeta}{m}-\varepsilon\ge\frac{\zeta}{m}$ and hence $\tau$ is not $(1-R+\frac{\zeta}{2m})$-implicitly rare.

By Theorem~\ref{thm:abundance}, a random linear code $\cC$ of rate $R$ contains some matrix $M\in\mathcal{M}_{n,\tau}$ with probability at least $1-n^{O_{L,q}(1)}q^{-\frac{\zeta n}{2m}}=1-q^{-\Omega_{p,q,\delta,\varepsilon}(n)}.$ By Proposition~\ref{prop:tau-bad}, the $L=m+q$ columns of $M$ form a bad list for $(p,L)$ average radius list decoding.

Finally, $L=m+q\ge \frac{p(1+h'_q(p))-2\zeta}{\varepsilon}-1+q.$ Since $2\zeta<\delta$, for sufficiently small $\varepsilon$ we have $L\ge \frac{p(1+h'_q(p))-\delta}{\varepsilon}.$ Therefore, $\cC$ is not
\[
\left(
p,
\left\lfloor
\frac{p(1+h'_q(p))-\delta}{\varepsilon}
\right\rfloor
\right)
\text{ average radius list decodable},
\]
which proves the theorem.
\end{proof}


Before proving Lemma~\ref{lem:not-implicitly-rare}, we first recall the following entropy bound from~\cite[Lemma~4.5]{GuruswamiEtAl2022}.

\begin{lemma}[{\cite[Lemma~4.5]{GuruswamiEtAl2022}}]
\label{lem:glmrsw-common-shift}
Let $q$ be a prime power, $p\in(0,1-1/q)$, $p'\in[p,1-1/q]$, and $\delta>0$. There exists $L_{p,q,\delta}$ such that, for all $L\ge L_{p,q,\delta}$ and $0\le d\le L$, the following holds. Let $\mathbf{v}$ be a vector sampled from $\mathrm{Bernoulli}_q(p)^d$ and let $\mathbf{w}$ be a vector independently sampled from $\mathrm{Bernoulli}_q(p')\cdot\mathbf{1}_d$. Then $H_q(\mathbf{v}+\mathbf{w})\ge d\left( h_q(p)+\frac{h_q(p)}{L+\delta}\right).$
\end{lemma}

\begin{rmk}
\label{rem:explicit-L}
Making the Hoeffding and Fano estimates in the proof of~\cite[Lemma~4.5]{GuruswamiEtAl2022} explicit gives the following sufficient choice. Set
\[
a_{p,q}:=\frac1q-\frac{p}{q-1}>0,
\qquad
\varepsilon_{p,q}
:=h_q\!\left(2p-\frac{q}{q-1}p^2\right)-h_q(p)>0.
\]
We may take
\[
d_{p,q,\delta}
:=
\left\lceil
\frac{2}{a_{p,q}^{\,2}}
\log\!\left[
\frac{3q}{h_q(p)}
\left(1+\frac{1}{\delta a_{p,q}^{\,2}}\right)
\right]
\right\rceil.
\]
Thus, we may choose
\[
L_{p,q,\delta}
:=
\left\lceil
\frac{d_{p,q,\delta}}{\varepsilon_{p,q}}
\right\rceil
=
O\!\left(
\frac{1}{\varepsilon_{p,q}a_{p,q}^{\,2}}
\log\!\left[
\frac{q}{h_q(p)}
\left(1+\frac{1}{\delta a_{p,q}^{\,2}}\right)
\right]
\right),
\]
Since $p_1(m)\to p\in(0,1-1/q)$, the quantities $a_{p_1(m),q}$, $\varepsilon_{p_1(m),q}$ and $h_q(p_1(m))$ are all bounded below by positive constants independent of $m$. Hence, $L_{p_1(m),q,1}=O_{p,q}(1).$ In particular, $m\ge L_{p_1(m),q,1}$ for all sufficiently large $m$, so Lemma~\ref{lem:glmrsw-common-shift} applies with $p=p'=p_1(m)$, $L=m$, and $\delta=1$.
\end{rmk}

We will use the following consequence of Lemma~\ref{lem:glmrsw-common-shift}.

\begin{cor}
\label{cor:ordinary-shift}
Fix any constant $c<h_q(p)$. For all sufficiently large $m$ and every $1\le d\le m$, let $\mathbf{v},\beta$ be two mutually independent variable such that 
\[
\mathbf{v}\sim\mathrm{Bernoulli}_q(p_1)^d,
\qquad
\beta\sim\mathrm{Bernoulli}_q(p_1).
\]
Then, for every $\mathbf{w}\in(\mathbb{F}_q^*)^d,$ we have $H_q(\mathbf{v}+\beta\mathbf{w})\ge d\,h_q(p_1)+c\frac{d}{m}.$
\end{cor}

\begin{proof}
Recall that $p_1=p+p/m-O_{p,q}(m^{-2})$. Thus, for all sufficiently large $m$, $p\le p_1\le \frac{p+1-1/q}{2}<1-\frac1q.$ By Remark~\ref{rem:explicit-L}, for all sufficiently large $m$, we may apply Lemma~\ref{lem:glmrsw-common-shift} with $p=p'=p_1$, $L=m$ and $\delta=1$. Let $D:=\mathrm{diag}(w_1^{-1},\ldots,w_d^{-1})$. Then $D\mathbf v$ has distribution $\mathrm{Bernoulli}_q(p_1)^d$ and is independent of $\beta$. Therefore,


\[
\begin{aligned}
H_q(\mathbf v+\beta\mathbf w)= H_q(D\mathbf v+\beta\mathbf 1_d)\ge d\left(h_q(p_1)+\frac{h_q(p_1)}{m+1}\right)\ge d h_q(p_1)+c\frac{d}{m},
\end{aligned}
\]
where the last inequality follows from $\frac{m h_q(p_1)}{m+1}\to h_q(p)>c$ when $m\to \infty$.
\end{proof}

The next lemma handles the case when $\beta\sim\mathrm{Unif}(\mathbb{F}_q)$. Compared with Corollary~\ref{cor:ordinary-shift}, the main improvement is that the constant $c$ can now be taken arbitrarily close to $1$, rather than being restricted by $c<h_q(p)$.

\begin{lemma}\label{lem:uniform-common-shift}
Fix any constant $c<1$. For all sufficiently large $m$ and every $1\le d\le m$, let $\mathbf{v},\beta$ be two mutually independent variables such that
\[
\mathbf{v}=(v_1,\ldots,v_d)
\sim \mathrm{Bernoulli}_q(p_1)^d,
\qquad
\beta\sim\mathrm{Unif}(\mathbb{F}_q).
\]
Then, for every $\mathbf{w}\in(\mathbb{F}_q^*)^d,$ we have $H_q(\mathbf{v}+\beta\mathbf{w})\ge d\,h_q(p_1)+c\frac{d}{m}.$
\end{lemma}

\begin{proof}
    The proof of this lemma is very similar to that of Lemma~\ref{lem:glmrsw-common-shift}, so we defer it to Appendix~\ref{app:proof-glmrsw-lemma45}.
\end{proof}

Before proving Lemma~\ref{lem:not-implicitly-rare}, it is convenient to represent the distribution $\tau$ with $m+2$ coordinates. Define $Z:=\alpha+B$ and $U_i:=E_i-B$ for $1\leq i\leq m$. Then the random vector $\mathbf{X}$ in Definition~\ref{def:bad-distribution} can be written as
\begin{equation}
\label{eq:intrinsic-representation}
\mathbf{X}=
\left(
Z+U_1,\ldots,Z+U_m,\,
(Z+aT)_{a\in\mathbb{F}_q}
\right).
\end{equation}

$\mathbf{X}$ is linearly equivalent to $\mathbf{W}:=(Z,U_1,\ldots,U_m,T)\in\mathbb{F}_q^{m+2}.$ To see this, the first $m+2$ coordinates of $\mathbf{X}$ is $Z+U_1,\ldots,Z+U_m, Z, Z+T$ which can clearly recover $Z,T,U_1,\ldots,U_m$. Moreover, it is clear that we can recover $\mathbf{X}$ from $\mathbf{W}$. Thus, they are linearly equivalent.


We will also use the fact that $Z\sim\mathrm{Unif}(\mathbb{F}_q)$ and $Z$ and $(U_1,\ldots,U_m,T)$ are mutually independent. This follows immediately from the fact $Z=\alpha+B$ as $\alpha$ is uniform and independent of $(E_1,\ldots,E_m,B,T)$. Moreover, every vector in $\mathbb{F}_q^{m+2}$ lies in the support set of $\mathbf{W}$. Indeed, if $T=0$, then $B=0$ and $E_1,\ldots,E_m$ can range over all value in $\F_q$. If $T\neq 0$,  the support set of $(B,E_1,\ldots,E_m)$ covers all vector in $\F_q^{m+1}$. In either case, $Z$ can take any value in $\mathbb{F}_q$. Therefore, for any linear map $A'$ acting on $\mathbf{W}$,
\[
\dim(A'\mathbf{W})
:=\dim\mathrm{span}(\mathrm{supp}(A'\mathbf{W}))
=\mathrm{rank}(A').
\]
To prove Lemma~\ref{lem:not-implicitly-rare}, it suffices to show that for every full-rank linear map $A':\mathbb{F}_q^{m+2}\to\mathbb{F}_q^r$, we have $H_q(A'\mathbf{W})\ge \gamma r$ where $\gamma$ is as defined in Lemma~\ref{lem:not-implicitly-rare}. Indeed, since $\mathbf{X}$ and $\mathbf{W}$ are linearly equivalent, for every linear transformation $A$ acting on $\mathbf{X}$ there is a linear map $A'$ acting on $\mathbf{W}$ such that $A\mathbf{X}=A'\mathbf{W}$. Hence, $\dim(A\tau)=\mathrm{rank}(A').$ We may assume that $A'$ has full row rank, since removing linearly dependent rows does not change the entropy of $A'\mathbf W$.


Combining Corollary~\ref{cor:ordinary-shift} and Lemma~\ref{lem:uniform-common-shift}, we now prove Lemma~\ref{lem:not-implicitly-rare}.

\begin{proof}[Proof of Lemma~\ref{lem:not-implicitly-rare}]
Write $\mathbf U:=(U_1,\ldots,U_m)$. We first show that every rank-$r$ linear map $A:\mathbb F_q^{m+1}\to\mathbb F_q^r$ satisfies $H_q(A(\mathbf U,T))\ge r\gamma$. The case $r=0$ is immediate. Write $A=[P\mid\mathbf t]$ and define $\mathbf Y:=P\mathbf U+\mathbf tT$. Since $A$ has rank $r$, we have $\mathrm{rank}(P)\in\{r,r-1\}$.

\medskip
\noindent\textbf{Case 1: $\mathrm{rank}(P)=r$.}

We want to prove the bound $H_q(\mathbf Y\mid T)\ge r\gamma$. Permuting the coordinates of $\mathbf U$ does not change their joint distribution with $T$, and invertible row operations preserve entropy. Thus, by Gaussian elimination, we may assume $P=[I_r\mid\mathbf w^{(1)}\mid\cdots\mid\mathbf w^{(k)}]$ with $k=m-r.$ Let $\mathbf V:=(E_1,\ldots,E_r)$ and $G_j:=E_{r+j}$ for $1\le j\le k$. Then
\[
\mathbf Y
=\mathbf V+\sum_{j=1}^k\mathbf w^{(j)}G_j
+\mathbf bB+\mathbf tT,
\qquad
\mathbf b:=-\mathbf1_r-\sum_{j=1}^k\mathbf w^{(j)}.
\]
For $1\le j\le k$, define $J_j:=\mathrm{supp}(\mathbf w^{(j)}) \setminus\bigcup_{a<j}J_a$ and $J_{k+1}:=[r]\setminus\bigcup_{j=1}^kJ_j.$ These sets partition $[r]$. For every $i\in J_j$, we have $w_i^{(j)}\neq0$ and $w_i^{(a)}=0$ for $a<j$. All the vectors $\mathbf w^{(1)},\ldots,\mathbf w^{(k)}$ are zero on $J_{k+1}$.

Fix $1\le j\le k$. $\mathbf V_{J_{j+1}},\ldots,\mathbf V_{J_{k+1}}, G_{j+1},\ldots,G_k,B,T$ completely determine  $\mathbf Y_{J_{j+1}},\ldots,\mathbf Y_{J_{k+1}}$. Thus, we can write $\mathbf Y_{J_j}=\mathbf V_{J_j}+\mathbf w^{(j)}_{J_j}G_j+\mathbf c$ where $\mathbf c$ is a linear combinations of $G_{j+1},\ldots,G_k,B,T$. Therefore, 

\[
\begin{aligned}
H_q\!\left(
\mathbf Y_{J_j}\,\middle|\,
\mathbf Y_{J_{j+1}},\ldots,\mathbf Y_{J_{k+1}},T
\right) &\ge
H_q\!\left(
\mathbf Y_{J_j}\,\middle|\,
\mathbf V_{J_{j+1}},\ldots,\mathbf V_{J_{k+1}},
G_{j+1},\ldots,G_k,B,T
\right) \\
&=
H_q\!\left(
\mathbf V_{J_j}+\mathbf w^{(j)}_{J_j}G_j
\right).
\end{aligned}
\]
Moreover, $\mathbf Y_{J_{k+1}}=\mathbf V_{J_{k+1}}-B\mathbf 1_{|J_{k+1}|}+\mathbf t_{J_{k+1}}T.$ Applying the chain rule to $\mathbf Y_{J_{k+1}},\mathbf Y_{J_k},\ldots,\mathbf Y_{J_1}$ gives
\[
\begin{aligned}
H_q(\mathbf Y\mid T)
&=
H_q(\mathbf Y_{J_{k+1}}\mid T)
+\sum_{j=1}^k
H_q\!\left(
\mathbf Y_{J_j}\,\middle|\,
\mathbf Y_{J_{j+1}},\ldots,\mathbf Y_{J_{k+1}},T
\right)\\
&=
H_q\!\left(
\mathbf V_{J_{k+1}}-B\mathbf 1_{|J_{k+1}|}
\,\middle|\,T
\right)
+\sum_{j=1}^k
H_q\!\left(
\mathbf Y_{J_j}\,\middle|\,
\mathbf Y_{J_{j+1}},\ldots,\mathbf Y_{J_{k+1}},T
\right)\\
&\ge
H_q\!\left(
\mathbf V_{J_{k+1}}-B\mathbf 1_{|J_{k+1}|}
\,\middle|\,T
\right)
+\sum_{j=1}^k
H_q\!\left(
\mathbf V_{J_j}
+\mathbf w^{(j)}_{J_j}G_j
\right).
\end{aligned}
\]
For all sufficiently large $m$, Taylor expansion \eqref{eq:Taylor} $\gamma-h_q(p_1)\le (p-\zeta/2)/m$. For every $1\le j\le k$, Corollary~\ref{cor:ordinary-shift} applied with $c=p-\zeta/4<h_q(p)$ gives
\[
H_q\!\left(
\mathbf V_{J_j}+\mathbf w^{(j)}_{J_j}G_j
\right)
\ge
|J_j|h_q(p_1)
+\left(p-\frac{\zeta}{4}\right)\frac{|J_j|}{m}
\ge |J_j|\gamma.
\]

It remains to consider $\mathbf Y_{J_{k+1}}$. When $T=0$, we have $B=0$; when $T\neq0$, $-B$ is uniform over $\mathbb F_q$ and independent of $\mathbf V_{J_{k+1}}$. Hence, applying Lemma~\ref{lem:uniform-common-shift} when $T\neq0$ with $c=1-\zeta/(4p)$, gives
\[
\begin{aligned}
H_q\!\left(
\mathbf V_{J_{k+1}}-B\mathbf1_{|J_{k+1}|}\mid T
\right)&=
(1-p_2)|J_{k+1}|h_q(p_1)
+p_2 H_q\!\left(
\mathbf V_{J_{k+1}}
-B\mathbf 1_{|J_{k+1}|}\mid T\neq0
\right) \\
&\ge
(1-p_2)|J_{k+1}|h_q(p_1)
+p_2\left(
|J_{k+1}|h_q(p_1)
+\left(1-\frac{\zeta}{4p}\right)
\frac{|J_{k+1}|}{m}
\right)\\
&=
|J_{k+1}|h_q(p_1)
+p_2\left(1-\frac{\zeta}{4p}\right)
\frac{|J_{k+1}|}{m}\ge |J_{k+1}|\gamma.
\end{aligned}
\]
The same bound is immediate when $|J_{k+1}|=0$. Combining the above bounds, we obtain $H_q(\mathbf Y\mid T)\ge\gamma\sum_{j=1}^{k+1}|J_j|=r\gamma.$ Thus, $H_q(\mathbf Y)\ge r\gamma$.

\medskip
\noindent\textbf{Case 2: $\mathrm{rank}(P)=r-1$.}

Since $\mathrm{rank}(P)=r-1$ while $\mathrm{rank}[P\mid \mathbf t]=r$, the vector $\mathbf t$ is not in the span of the columns of $P$. Hence, after invertible row operations, we may write
\[
P=
\begin{pmatrix}
P_0\\
0
\end{pmatrix},
\qquad
\mathbf t=
\begin{pmatrix}
0\\
1
\end{pmatrix},
\]
where $\mathrm{rank}(P_0)=r-1$. Therefore, after invertible linear transformation, $A(\mathbf U,T)=(P_0\mathbf U,T)$ and not change entropy. Thus, the conditional entropy bound from Case~1 applies to $P_0$, we obtain
\[
H_q(A(\mathbf U,T))
=
H_q(T)+H_q(P_0\mathbf U\mid T)
\ge h_q(p_2)+(r-1)\gamma
\ge r\gamma
\]
as \eqref{eq:gamma} indicates $\gamma\leq h_q(p_2)$ for large enough $m$.

Finally, we add back $Z$, which is uniform and independent of $(\mathbf U,T)$. Let $A'=[\mathbf a\mid A_0]$ be a rank-$r$ linear map on $\mathbf W=(Z,\mathbf U,T)$. If $\mathrm{rank}(A_0)=r$, then $H_q(A'\mathbf W)\ge H_q(A'\mathbf W\mid Z)=H_q(A_0(\mathbf U,T))\ge r\gamma.$ Otherwise, $\mathrm{rank}(A_0)=r-1$. Since $A'=[\mathbf a\mid A_0]$ has rank $r$, the vector $\mathbf a$ is not in the span of the columns of $A_0$. Hence, after invertible row operations, we may write $A'\mathbf W=(A_1(\mathbf U,T),Z)$ where $\mathrm{rank}(A_1)=r-1$. Since $Z$ is uniform and independent of $(\mathbf U,T)$, $H_q(A'\mathbf W)=H_q(A_1(\mathbf U,T))+1\ge (r-1)\gamma+1\ge r\gamma.$ Therefore, for every full-rank linear transformation $A$ acting on $\mathbf X$, we obtain $H_q(A\tau)\ge \gamma\,\dim(A\tau).$ Hence, $\tau$ is not $\gamma$-implicitly rare.
\end{proof}

\begin{rmk}
\label{rem:improvement-range}
For every prime power $q\ge 3$, let $p_q$ be the unique solution of $q^p(1-p)=1$ in $(0,1-1/q)$. Then our leading constant exceeds $h_q(p)$ exactly when $0<p<p_q$. For $q=2$, no such interval exists. Some numerical values are listed below.

\begin{table}[H]
\centering
\setlength{\belowcaptionskip}{6pt}
\caption{Numerical values of $p_q$.}
\label{tab:numerical-pq}

\renewcommand{\arraystretch}{1.3}
\setlength{\tabcolsep}{9pt}
\setlength{\arrayrulewidth}{0.45pt}
\arrayrulecolor{pqrule}

\begin{tabular}{|ccc|ccc|}
\hline
\rowcolor{pqhead}
\makebox[1.4em]{$q$} & $p_q$ & $1-1/q$ &
\makebox[1.4em]{$q$} & $p_q$ & $1-1/q$ \\ \hline
$2$ & $-$      & $0.5000$ & $8$  & $0.8172$ & $0.8750$ \\
$3$ & $0.1740$ & $0.6667$ & $9$  & $0.8432$ & $0.8889$ \\
$4$ & $0.5000$ & $0.7500$ & $11$ & $0.8783$ & $0.9091$ \\
$5$ & $0.6470$ & $0.8000$ & $13$ & $0.9008$ & $0.9231$ \\
$7$ & $0.7814$ & $0.8571$ & $16$ & $0.9225$ & $0.9375$ \\ \hline
\end{tabular}

\arrayrulecolor{black}
\end{table}
\end{rmk}

\section{A Stronger Average-Radius Lower Bound}\label{sec:average-multiline}

Throughout this section, fix a prime power $q$ and $0<p<1-1/q$. We extend the construction in Section~\ref{sec:average}. For every integer $s\ge1$, define
\begin{equation}
\label{ml:eq:leading-constant}
K_s:=p h_q'(p)+p^s\left(1+(s-1)\left(1-\frac1q\right)h_q'(p)\right).
\end{equation}
Our main theorem in this section is the following.

\begin{thm}
\label{ml:thm:lower bound}
Fix a prime power $q$, $p\in(0,1-1/q)$, an integer $s\ge1$ and $\delta\in(0,K_s)$. There exists $\varepsilon_{p,q,s,\delta}>0$ such that, for every $\varepsilon\in(0,\varepsilon_{p,q,s,\delta})$ and all sufficiently large $n$, a random linear code $\cC\subseteq\mathbb F_q^n$ of rate $1-h_q(p)-\varepsilon$ is not
\[
\left(p,\left\lfloor\frac{K_s-\delta}{\varepsilon}\right\rfloor\right)
\text{-average-radius list-decodable}
\]
with probability $1-q^{-\Omega_{p,q,s,\delta,\varepsilon}(n)}$.
\end{thm}

Our proof follows the same argument as the proof of Theorem~\ref{thm:average radius-lower bound}. We first define a distribution $\tau$ in Definition~\ref{ml:def:distribution} and show that it is \emph{bad} in Proposition~\ref{ml:prop:bad}. Next, we prove that $\tau$ is not $\gamma$-implicitly rare in Lemma~\ref{ml:lem:not-rare}, and hence is \emph{abundant} by Theorem~\ref{thm:abundance}. We then use these bounds to prove Theorem~\ref{ml:thm:lower bound}. All proofs are deferred to Appendix~\ref{app:multiline}.

Fix a sufficiently large integer $m$ and set $L:=m+1+s(q-1)$. Define
\begin{equation}
\label{ml:eq:noise-parameters}
p_1:=p+\frac{K_s-p^s}{h_q'(p)m}
-\frac{s(q-1)(K_s+2)+h_q'(p)}{h_q'(p)m^2},
\qquad
p_2:=p+\frac{K_s+2}{h_q'(p)m}.
\end{equation}
and the parameters belong to $(0,1-1/q)$.

\begin{defn}[The bad distribution $\tau$]
\label{ml:def:distribution}
Let
\[
\alpha\sim\mathrm{Unif}(\mathbb F_q),\qquad
E_1,\ldots,E_m\stackrel{\mathrm{i.i.d.}}{\sim}\mathrm{Bernoulli}_q(p_1),\qquad
T_1,\ldots,T_s\stackrel{\mathrm{i.i.d.}}{\sim}\mathrm{Bernoulli}_q(p_2),
\]
where $\alpha,E_1,\ldots,E_m,T_1,\ldots,T_s$ are mutually independent. Conditioned on $T_1,\ldots,T_s$, define a random variable $B$ as follows: if at least one $T_j$ is zero, let $B=0$; if every $T_j$ is nonzero, let $B$ be uniformly distributed over $\mathbb F_q$, independently of $\alpha,E_1,\ldots,E_m$. Define $\tau$ on $\mathbb F_q^L$ as the distribution of the random vector
\[
\mathbf X=\left(
\alpha+E_1,\ldots,\alpha+E_m,\,
\alpha+B,\,
(\alpha+B+aT_j)_{1\le j\le s,\ a\in\mathbb F_q^*}
\right).
\]
\end{defn}

As in Definition~\ref{def:bad-distribution}, the first $m$ coordinates are $\alpha+E_1,\ldots,\alpha+E_m$. For each $j\in[s]$, we also include the coordinates $\alpha+B+aT_j$ for $a\in\mathbb F_q$, with the common coordinate $\alpha+B$ corresponding to $a=0$ included only once. Thus, there are $1+s(q-1)$ remaining coordinates. We first show that $\tau$ is bad for average-radius list decoding.

\begin{prop}
\label{ml:prop:bad}
Let $\tau$ be as in Definition~\ref{ml:def:distribution}. Let $\cC\subseteq\mathbb F_q^n$ and $M\in\mathcal M_{n,\tau}$. If $M\subseteq\cC$, then $\cC$ is not $(p,L)$-average-radius list-decodable for all sufficiently large $n$.
\end{prop}

Next, we show that $\tau$ is not implicitly rare, which will imply that $\tau$ is abundant by Theorem~\ref{thm:abundance}.

\begin{lemma}
\label{ml:lem:not-rare}
Let $p\in(0,1-1/q)$, let $s\ge1$ be an integer and let $0<\zeta<\min\{p^s/4,K_s/4\}$. There exists $m_{p,q,s,\zeta}>0$ such that, for every $m\ge m_{p,q,s,\zeta}$, the distribution $\tau$ given in Definition~\ref{ml:def:distribution} is not $\gamma$-implicitly rare where $\gamma:=h_q(p)+\frac{K_s-\zeta}{m}$.
\end{lemma}

Before proving Lemma~\ref{ml:lem:not-rare}, we rewrite the distribution $\tau$ using $m+s+1$ coordinates as in Section~\ref{sec:average}. Define $Z:=\alpha+B$, $U_i:=E_i-B$ for $1\le i\le m$, $\mathbf U:=(U_1,\ldots,U_m)$ and $\mathbf T:=(T_1,\ldots,T_s)$. Then
\[
\mathbf X=\left(
Z+U_1,\ldots,Z+U_m,\,Z,\,
(Z+aT_j)_{1\le j\le s,\ a\in\mathbb F_q^*}
\right).
\]
From $\mathbf X$, we can recover $Z$ directly, recover each $T_j$ by subtracting $Z$ from the coordinate $Z+T_j$ and then recover $U_i=X_i-Z$. Thus, $\mathbf{X}$ is linearly equivalent to $\mathbf{W}:=(Z,U_1,\ldots,U_m,T)\in\mathbb{F}_q^{m+s+1}.$

We will also use the fact that $Z\sim\mathrm{Unif}(\mathbb F_q)$ and $Z$ and $(\mathbf U,\mathbf T)$ are mutually independent. This follows from $Z=\alpha+B$, since $\alpha$ is uniform and independent of $(E_1,\ldots,E_m,B,\mathbf T)$. Moreover, every vector in $\mathbb F_q^{m+s+1}$ occurs with positive probability under $\mathbf W$. Indeed, for any fixed value of $\mathbf T$, we may take $B=0$ with positive probability and then choose $\alpha,E_1,\ldots,E_m$ to obtain any prescribed values of $Z,U_1,\ldots,U_m$. 

Consequently, for every linear map $A$ acting on $\mathbf X$, there is a linear map $A'$ acting on $\mathbf W$ such that $A\mathbf X=A'\mathbf W$ and $\dim(A\tau)=\dim (A'\mathbf{W})=\dim\mathrm{span}(\mathrm{supp}(A'\mathbf{W}))=\mathrm{rank}(A').$ Removing linearly dependent coordinates does not change entropy. Hence, to prove Lemma~\ref{ml:lem:not-rare}, it therefore suffices to show that every full-rank linear map $A':\mathbb F_q^{m+s+1}\to\mathbb F_q^r$ satisfies $H_q(A'\mathbf W)\ge r\gamma.$

We will use Corollary~\ref{cor:ordinary-shift} and Lemma~\ref{lem:uniform-common-shift} with $p_1$ defined in~\eqref{ml:eq:noise-parameters}. Their proofs only require $p_1\to p$, so the same bounds hold for this choice of $p_1$ when $m$ is sufficiently large.

The next lemma is the analogue of the estimate used for the last coordinates in the proof of Lemma~\ref{lem:not-implicitly-rare}. The difference is that we now have several variables $T_1,\ldots,T_s$, so we need to bound an entropy term of the form $H_q\!\left(A\mathbf T,\,(E_1,\ldots,E_d)-B\mathbf 1_d+Q\mathbf T\right).$ The lemma shows that the additional linear terms $A\mathbf T$ and $Q\mathbf T$ do not affect the lower bound obtained from $B$. Here $\zeta$ and $\gamma$ are as in Lemma~\ref{ml:lem:not-rare}.

\begin{lemma}
\label{ml:lem:last-block}
Using the above notations, for all sufficiently large $m$, let $0\le d\le m$, $A:\mathbb F_q^s\to\mathbb F_q^j$ be a linear map of rank $j$ where $0\le j\le s$ and let $Q\in\mathbb F_q^{d\times s}$. Then
\[
H_q\!\left(
A\mathbf T,
(E_1,\ldots,E_d)-B\mathbf1_d+Q\mathbf T
\right)
\ge j\gamma+d h_q(p_1)
+\left(p^s-\frac{\zeta}{4}\right)\frac{d}{m}
\ge(j+d)\gamma.
\]
\end{lemma}

\begin{rmk}
When $s=1$, we have $K_1=p(1+h_q'(p))$, so Theorem~\ref{ml:thm:lower bound} recovers Theorem~\ref{thm:average radius-lower bound}. For fixed $p,q$, we have $K_{s+1}-K_s=(1-p)\left(1-\frac1q\right)h_q'(p)p^s\left(\frac1{1-p}-\frac1{(1-1/q)h_q'(p)}-s\right).$ Thus, $K_{s_*}=\max_{s\ge1}K_s$ for $s_*:=\max\left\{1,\,\left\lceil\frac1{1-p}-\frac1{(1-1/q)h_q'(p)}\right\rceil\right\}.$ In particular, for each fixed $q$, $\max_{s\ge1}K_s>h_q(p)$ exactly when $q^p(1-p)>1$, so the range of $p$ is the same as in Remark~\ref{rem:improvement-range}.

For $p=1-\lambda/q$ with $1<\lambda<q$, this becomes $s_*=\max\left\{1,\,\left\lceil\frac q\lambda-\frac{q}{(q-1)\log_q\!\left(\frac{\lambda(q-1)}{q-\lambda}\right)}\right\rceil\right\}.$ For fixed $\lambda>1$ and $q\to\infty$, substituting into~\eqref{ml:eq:leading-constant} gives $K_{s_*}=\left(\frac{\ln\lambda}{e\lambda}+o_q(1)\right)\frac q{\ln q}.$
\end{rmk}

\section{Upper Bounds for Average Radius List Decoding}\label{sec:average-upper}
Throughout this section, fix a prime power $q$ and $0<p<1-1/q$. We prove the following upper bound for average radius list decoding.

\begin{thm}
\label{thm:ar-upper}
Fix $0<\varepsilon<1-h_q(p)$. For all sufficiently large $n$, a random linear code $\cC\subseteq\mathbb{F}_q^n$ of dimension $\lfloor(1-h_q(p)-\varepsilon)n\rfloor$ is $(p,L)$ average radius list decodable with probability at least $1-q^{-\Omega_{p,q,\varepsilon}(n)}$, where
\[
L=\left\lfloor\frac{q-1}{q^{\varepsilon/(h_q(p)+\varepsilon)}-1}\right\rfloor+1.
\]
In particular, as $\varepsilon\to0$, we have $L=\frac{(q-1)h_q(p)}{\varepsilon\ln q}+O_{p,q}(1)$.
\end{thm}

We first recall two results used in the proof. Let $P$ be a nonzero polynomial in $m$ variables over a field $\mathbb{F}$. For $\mathbf{z}\in\mathbb{F}^m$, define the multiplicity $\mathrm{mult}(P,\mathbf{z})$ of $P$ at $\mathbf{z}$ as the minimum total degree of a monomial with a nonzero coefficient in $P(\mathbf{z}+\mathbf{Z})$, where $\mathbf{Z}:=(Z_1,\ldots,Z_m)^{\mathsf T}$.

\begin{lemma}[{\cite[Lemma~2.7]{DvirKoppartySarafSudan2013}}]
\label{lem:ar-multiplicity}
Let $P\in\mathbb{F}[Z_1,\ldots,Z_m]$ be a nonzero polynomial of total degree at most $d$. For every nonempty finite set $S\subseteq\mathbb{F}$, we have
\[
\sum_{\mathbf{z}\in S^m}\mathrm{mult}(P,\mathbf{z})\le d|S|^{m-1}.
\]
\end{lemma}

Since every subcode of a $(p,L)$ average radius list decodable code is also $(p,L)$ average radius list decodable, Rossman's threshold theorem gives the following lemma.

\begin{lemma}[{\cite[Theorem~2]{Rossman2019}}]
\label{lem:ar-threshold}
Let $L\ge2$ and $0\le k<k'<n$ be integers. Suppose that a random linear code in $\mathbb{F}_q^n$ of dimension $k'$ is $(p,L)$ average radius list decodable with probability at least $c\in(0,1)$. Then a random linear code in $\mathbb{F}_q^n$ of dimension $k$ is $(p,L)$ average radius list decodable with probability at least $1-\frac{1-c}{c}q^{-(k'-k)}$.
\end{lemma}

Fix $0<\varepsilon<1-h_q(p)$ and let $k:=\lfloor(1-h_q(p)-\varepsilon)n\rfloor$ and $r:=n-k$. A uniformly random matrix in $\mathbb{F}_q^{r\times n}$ has rank $r$ with probability at least $1-q^{-k}$. We therefore choose $H\in\mathbb{F}_q^{r\times n}$ uniformly at random among matrices of rank $r$ and let $\cC:=\ker H$. Then $\cC$ is a random linear code of dimension $k$. For now, fix an integer $L\ge2$ independent of $n$.

Let $\tau:=\mathrm{Bernoulli}_q(p)^n$. For every $\mathbf{x}\in\mathbb{F}_q^n$, we have
\[
\tau(\mathbf{x})=(1-p)^{n-\mathrm{wt}(\mathbf{x})}\left(\frac{p}{q-1}\right)^{\mathrm{wt}(\mathbf{x})}.
\]

\begin{lemma}
\label{lem:ar-reformulation}
Let $H\in\mathbb{F}_q^{r\times n}$ and let $L\ge2$ be an integer. Then $\ker H$ is not $(p,L)$ average radius list decodable if and only if there exist $L$ distinct vectors $\mathbf{x}_1,\ldots,\mathbf{x}_L\in\mathbb{F}_q^n$ such that $H\mathbf{x}_1=\cdots=H\mathbf{x}_L$ and $\prod_{i=1}^L\tau(\mathbf{x}_i)\ge q^{-Lnh_q(p)}$.
\end{lemma}

\begin{proof}
By the definition of $\tau$, for any $\mathbf{x}_1,\ldots,\mathbf{x}_L\in\mathbb{F}_q^n$, we have
\[
\prod_{i=1}^L\tau(\mathbf{x}_i)
=(1-p)^{Ln}\left(\frac{p}{(q-1)(1-p)}\right)^{\sum_{i=1}^L\mathrm{wt}(\mathbf{x}_i)}.
\]
Since $p/(q-1)<1-p$, the inequality $\sum_{i=1}^L\mathrm{wt}(\mathbf{x}_i)\le Lpn$ holds if and only if
\[
\prod_{i=1}^L\tau(\mathbf{x}_i)
\ge (1-p)^{(1-p)Ln}\left(\frac{p}{q-1}\right)^{pLn}
=q^{-Lnh_q(p)}.
\]

Suppose that $\ker H$ is not $(p,L)$ average radius list decodable. Then there exist $L$ distinct codewords $\mathbf{c}_1,\ldots,\mathbf{c}_L\in\ker H$ and a vector $\mathbf{y}\in\mathbb{F}_q^n$ such that $\sum_{i=1}^L d(\mathbf{c}_i,\mathbf{y})\le Lpn$. Let $\mathbf{x}_i:=\mathbf{c}_i-\mathbf{y}$ for each $i\in[L]$. Then $\mathbf{x}_1,\ldots,\mathbf{x}_L$ are distinct and $H\mathbf{x}_i=-H\mathbf{y}$ for each $i\in[L]$. Since $\sum_{i=1}^L\mathrm{wt}(\mathbf{x}_i)=\sum_{i=1}^L d(\mathbf{c}_i,\mathbf{y})\le Lpn$, the above equivalence gives $\prod_{i=1}^L\tau(\mathbf{x}_i)\ge q^{-Lnh_q(p)}$.

Conversely, suppose that $\mathbf{x}_1,\ldots,\mathbf{x}_L$ satisfy the conditions in the lemma. Let $\mathbf{c}_i:=\mathbf{x}_i-\mathbf{x}_1$ for each $i\in[L]$ and let $\mathbf{y}:=-\mathbf{x}_1$. Then $\mathbf{c}_1,\ldots,\mathbf{c}_L$ are distinct codewords in $\ker H$. By the above equivalence, $\sum_{i=1}^L d(\mathbf{c}_i,\mathbf{y})=\sum_{i=1}^L\mathrm{wt}(\mathbf{x}_i)\le Lpn$. Thus $\ker H$ is not $(p,L)$ average radius list decodable.
\end{proof}

For each $\mathbf{x}\in\mathbb{F}_q^n$, define $f(\mathbf{x}):=(1+q^r\tau(\mathbf{x}))^{1/n}$. Since $\tau$ is a probability distribution, we have
\[
\sum_{\mathbf{x}\in\mathbb{F}_q^n}f(\mathbf{x})^n
=\sum_{\mathbf{x}\in\mathbb{F}_q^n}(1+q^r\tau(\mathbf{x}))=q^n+q^r.
\]

Suppose that $\cC$ is not $(p,L)$ average radius list decodable. By Lemma~\ref{lem:ar-reformulation}, there exist $L$ distinct vectors $\mathbf{x}_1,\ldots,\mathbf{x}_L\in\mathbb{F}_q^n$ and a vector $\mathbf{b}\in\mathbb{F}_q^r$ such that $H\mathbf{x}_i=\mathbf{b}$ for each $i\in[L]$ and $\prod_{i=1}^L\tau(\mathbf{x}_i)\ge q^{-Lnh_q(p)}$. Since $f(\mathbf{x})\ge1$ and $f(\mathbf{x})\ge q^{r/n}\tau(\mathbf{x})^{1/n}$ for each $\mathbf{x}\in\mathbb{F}_q^n$, we have
\begin{align*}
\sum_{\mathbf{x}:H\mathbf{x}=\mathbf{b}}(f(\mathbf{x})-1)
&\ge \sum_{i=1}^L f(\mathbf{x}_i)-L\ge Lq^{r/n}\left(\prod_{i=1}^L\tau(\mathbf{x}_i)\right)^{1/(Ln)}-L\\
&\ge L(q^{r/n-h_q(p)}-1)\ge L(q^\varepsilon-1),
\end{align*}
where the second inequality follows from the AM-GM inequality and the last follows from $r\ge(h_q(p)+\varepsilon)n$. By Markov's inequality, the probability that $\cC$ is not $(p,L)$ average radius list decodable is at most
\[
\frac{1}{L(q^\varepsilon-1)}\mathbb{E}_H\left[\max_{\mathbf{b}\in\mathbb{F}_q^r}\sum_{\mathbf{x}:H\mathbf{x}=\mathbf{b}}(f(\mathbf{x})-1)\right].
\]

To bound this expectation, we prove the following lemma using the polynomial method in~\cite[Lemma~2.9 and Proposition~2.11]{KM26}.

\begin{lemma}
\label{lem:ar-sum}
Let $N>r\ge1$ be integers and let $f:\mathbb{F}_q^N\to\mathbb{R}$ satisfy $f(\mathbf{z})\ge1$ for every $\mathbf{z}\in\mathbb{F}_q^N$ and $\sum_{\mathbf{z}\in\mathbb{F}_q^N}f(\mathbf{z})^N\le q^N+q^r$. Choose $B\in\mathbb{F}_q^{r\times N}$ uniformly at random among matrices of rank $r$. Then
\[
\mathbb{E}_B\left[\max_{\mathbf{b}\in\mathbb{F}_q^r}\sum_{\mathbf{z}:B\mathbf{z}=\mathbf{b}}(f(\mathbf{z})-1)\right]
\le q-1+O_q\left(\frac1N+q^{-r}\right).
\]
\end{lemma}

\begin{proof}
Let $t:=N-r$ and fix a positive integer $u$. Let $D:=\lceil q(u+N)(1+q^{-t})^{1/N}\rceil$. We first show that there exists a nonzero polynomial $P\in\mathbb{F}_q[Z_1,\ldots,Z_N]$ of degree at most $D$ such that $\mathrm{mult}(P,\mathbf{z})\ge uf(\mathbf{z})$ for every $\mathbf{z}\in\mathbb{F}_q^N$.

For each $\mathbf{z}\in\mathbb{F}_q^N$, the coefficients of all monomials of total degree less than $\lceil uf(\mathbf{z})\rceil$ in $P(\mathbf{z}+\mathbf{Z})$ must be zero. This gives $\binom{\lceil uf(\mathbf{z})\rceil+N-1}{N}$ homogeneous linear equations in the coefficients of $P$. Since $f(\mathbf{z})\ge1$, we have $\lceil uf(\mathbf{z})\rceil+N-1\le (u+N)f(\mathbf{z})$. Thus, the total number of equations is at most
\[
\sum_{\mathbf{z}\in\mathbb{F}_q^N}\binom{\lceil uf(\mathbf{z})\rceil+N-1}{N}
\le \frac{(u+N)^N}{N!}\sum_{\mathbf{z}\in\mathbb{F}_q^N}f(\mathbf{z})^N
\le \frac{(u+N)^N(q^N+q^r)}{N!}.
\]
There are $\binom{D+N}{N}>D^N/N!$ monomials in $N$ variables of total degree at most $D$. By the choice of $D$, the number of unknown coefficients of $P$ is greater than the total number of equations. Hence these equations have a nonzero solution, giving the required polynomial $P$. Let $d:=\deg P$ and write $P_d$ for the homogeneous part of $P$ of degree $d$.

Fix $B$ and let $G\in\mathbb{F}_q^{N\times t}$ be a matrix whose columns form a basis of $\ker B$. For each $\mathbf{b}\in\mathbb{F}_q^r$, choose $\mathbf{z}_0\in\mathbb{F}_q^N$ such that $B\mathbf{z}_0=\mathbf{b}$. Then every solution of $B\mathbf{z}=\mathbf{b}$ is of the form $\mathbf{z}_0+G\mathbf{v}$ for a unique $\mathbf{v}\in\mathbb{F}_q^t$. Write $\mathbf{V}:=(V_1,\ldots,V_t)^{\mathsf T}$ and $\mathbf{Z}:=(Z_1,\ldots,Z_N)^{\mathsf T}$.

To apply Lemma~\ref{lem:ar-multiplicity}, we need a nonzero polynomial in $\mathbf{V}$. Since $P(\mathbf{z}_0+G\mathbf{V})$ may be the zero polynomial, we instead choose a nonzero coefficient of a monomial in $\mathbf{Z}$ in the expansion of $P(\mathbf{z}_0+G\mathbf{V}+\mathbf{Z})$.

Let $a_1,\ldots,a_t$ be a basis of $\mathbb{F}_{q^t}$ over $\mathbb{F}_q$ and let $\mathbf{a}:=(a_1,\ldots,a_t)^{\mathsf T}$. Viewing $P_d$ as a polynomial over $\mathbb{F}_{q^t}$, let $\ell:=\mathrm{mult}(P_d,G\mathbf{a})$. By the definition of multiplicity, there exist nonnegative integers $j_1,\ldots,j_N$ with $\sum_{i=1}^N j_i=\ell$ such that the coefficient of $Z_1^{j_1}\cdots Z_N^{j_N}$ in $P_d(G\mathbf{a}+\mathbf{Z})$ is nonzero.

For each $\mathbf{z}_0\in\mathbb{F}_q^N$, let $P_{\mathbf{z}_0}(\mathbf{V})$ denote the coefficient of $Z_1^{j_1}\cdots Z_N^{j_N}$ in $P(\mathbf{z}_0+G\mathbf{V}+\mathbf{Z})$. The homogeneous part of $P(\mathbf{z}_0+G\mathbf{V}+\mathbf{Z})$ of total degree $d$ in $\mathbf{V}$ and $\mathbf{Z}$ is $P_d(G\mathbf{V}+\mathbf{Z})$. Thus $P_{\mathbf{z}_0}$ has degree at most $d-\ell$, and its homogeneous part of degree $d-\ell$ is the coefficient of $Z_1^{j_1}\cdots Z_N^{j_N}$ in $P_d(G\mathbf{V}+\mathbf{Z})$. By the choice of $j_1,\ldots,j_N$, this coefficient is nonzero at $\mathbf{V}=\mathbf{a}$. Therefore, $P_{\mathbf{z}_0}$ is nonzero and has degree $d-\ell$ for every $\mathbf{z}_0\in\mathbb{F}_q^N$.

We now bound the multiplicity of $P_{\mathbf{z}_0}$ at each $\mathbf{v}\in\mathbb{F}_q^t$. Since $\mathrm{mult}(P,\mathbf{z}_0+G\mathbf{v})\ge uf(\mathbf{z}_0+G\mathbf{v})$, the polynomial $P(\mathbf{z}_0+G\mathbf{v}+G\mathbf{V}+\mathbf{Z})$ contains no term of total degree less than $uf(\mathbf{z}_0+G\mathbf{v})$ in $\mathbf{V}$ and $\mathbf{Z}$. The coefficient of $Z_1^{j_1}\cdots Z_N^{j_N}$ in this polynomial is $P_{\mathbf{z}_0}(\mathbf{v}+\mathbf{V})$. Since $\sum_{i=1}^N j_i=\ell$, we have $\mathrm{mult}(P_{\mathbf{z}_0},\mathbf{v})\ge uf(\mathbf{z}_0+G\mathbf{v})-\ell.$ Summing over $\mathbf{v}\in\mathbb{F}_q^t$ and applying Lemma~\ref{lem:ar-multiplicity} to $P_{\mathbf{z}_0}$, we obtain
\[
u\sum_{\mathbf{z}:B\mathbf{z}=\mathbf{b}}f(\mathbf{z})-q^t\ell
\le \sum_{\mathbf{v}\in\mathbb{F}_q^t}\mathrm{mult}(P_{\mathbf{z}_0},\mathbf{v})
\le (d-\ell)q^{t-1}.
\]
Since $B\mathbf{z}=\mathbf{b}$ has exactly $q^t$ solutions, it follows that
\[
\sum_{\mathbf{z}:B\mathbf{z}=\mathbf{b}}(f(\mathbf{z})-1)
\le q^t\left(\frac{d}{qu}-1\right)+\frac{(q-1)q^{t-1}\ell}{u}.
\]
Since $\ell=\mathrm{mult}(P_d,G\mathbf{a})$ is independent of $\mathbf{b}$, the same upper bound holds for the maximum over all $\mathbf{b}\in\mathbb{F}_q^r$.

It remains to bound the expectation of $\ell$. Choose $G$ uniformly at random in $\mathbb{F}_q^{N\times t}$ and let $\ell=\mathrm{mult}(P_d,G\mathbf{a})$ as above. Since $a_1,\ldots,a_t$ form a basis of $\mathbb{F}_{q^t}$ over $\mathbb{F}_q$, each coordinate of $G\mathbf{a}$ is uniform over $\mathbb{F}_{q^t}$ and the coordinates are independent. By Lemma~\ref{lem:ar-multiplicity} with $S=\mathbb{F}_{q^t}$, we have $\mathbb{E}_G[\ell]\le d/q^t$. Moreover, $\Pr[\mathrm{rank}(G)=t]\ge1-q^{-r}$. Since $\ell\ge0$, we have $\mathbb{E}_G[\ell\mid\mathrm{rank}(G)=t]\le d/(q^t(1-q^{-r}))$.

Given $\mathrm{rank}(G)=t$, the column space of $G$ has the same distribution as $\ker B$, namely the uniform distribution on the subspaces of $\mathbb{F}_q^N$ of dimension $t$. As $\mathbf{b}$ varies over $\mathbb{F}_q^r$, the solution sets of $B\mathbf{z}=\mathbf{b}$ are exactly the translates of $\ker B$. Thus the maximum over $\mathbf{b}$ depends only on $\ker B$. Taking expectations gives
\[
\mathbb{E}_B\left[\max_{\mathbf{b}\in\mathbb{F}_q^r}\sum_{\mathbf{z}:B\mathbf{z}=\mathbf{b}}(f(\mathbf{z})-1)\right]
\le q^t\left(\frac{d}{qu}-1\right)+\frac{(q-1)d}{qu(1-q^{-r})}.
\]
Since $d\le D=\lceil q(u+N)(1+q^{-t})^{1/N}\rceil$ and the expectation on the left hand side is independent of $u$, we let $u\to\infty$. Using $(1+q^{-t})^{1/N}\le1+q^{-t}/N$, we obtain
\begin{align*}
&\mathbb{E}_B\left[\max_{\mathbf{b}\in\mathbb{F}_q^r}\sum_{\mathbf{z}:B\mathbf{z}=\mathbf{b}}(f(\mathbf{z})-1)\right]\\
&\qquad\le q^t\left((1+q^{-t})^{1/N}-1\right)+\frac{(q-1)(1+q^{-t})^{1/N}}{1-q^{-r}}\\
&\qquad\le \frac1N+\frac{(q-1)(1+q^{-t}/N)}{1-q^{-r}}=q-1+O_q\left(\frac1N+q^{-r}\right).\qedhere
\end{align*}
\end{proof}

Applying Lemma~\ref{lem:ar-sum} with $N=n$ and $B=H$, we have
\[
\mathbb{E}_H\left[\max_{\mathbf{b}\in\mathbb{F}_q^r}\sum_{\mathbf{x}:H\mathbf{x}=\mathbf{b}}(f(\mathbf{x})-1)\right]\le q-1+o(1).
\]
Thus the probability that $\cC$ is not $(p,L)$ average radius list decodable is at most $(q-1+o(1))/(L(q^\varepsilon-1))$. In particular, for $L=\lfloor(q-1)/(q^\varepsilon-1)\rfloor+1$, the probability that $\cC$ is $(p,L)$ average radius list decodable is bounded below by a positive constant for all sufficiently large $n$.

To obtain a better bound on $L$, we write $H=BA$, where $A\in\mathbb{F}_q^{N\times n}$, $B\in\mathbb{F}_q^{r\times N}$ and $N=r+\lfloor\sqrt n\rfloor$. Applying Lemma~\ref{lem:ar-sum} to the function $f(\mathbf{z}):=(1+q^r(A\tau)(\mathbf{z}))^{1/N}$ on $\mathbb{F}_q^N$, we can replace the exponent $(r-nh_q(p))/n$ in the above argument by $(r-nh_q(p))/N$.

\begin{prop}
\label{prop:ar-positive-probability}
Fix $0<\varepsilon<1-h_q(p)$ and an integer $L\ge2$ such that $L(q^{\varepsilon/(h_q(p)+\varepsilon)}-1)>q-1$. There exists a constant $c=c(p,q,\varepsilon,L)>0$ such that, for all sufficiently large $n$, a random linear code in $\mathbb{F}_q^n$ of dimension $\lfloor(1-h_q(p)-\varepsilon)n\rfloor$ is $(p,L)$ average radius list decodable with probability at least $c$.
\end{prop}

\begin{proof}
Let $k:=\lfloor(1-h_q(p)-\varepsilon)n\rfloor$, $r:=n-k$, $t:=\lfloor\sqrt n\rfloor$ and $N:=r+t$. For all sufficiently large $n$, we have $r<N<n$. Choose $A\in\mathbb{F}_q^{N\times n}$ and $B\in\mathbb{F}_q^{r\times N}$ independently and uniformly at random among matrices of rank $N$ and $r$, respectively, and let $H:=BA$.

We first show that $H$ is uniform among matrices in $\mathbb{F}_q^{r\times n}$ of rank $r$ and that, given $H$, the kernel of $A$ is a uniformly random subspace of $\ker H$ of dimension $n-N$. Fix a matrix $H$ of rank $r$ and a subspace $D\subseteq\ker H$ of dimension $n-N$. For every matrix $A$ of rank $N$ with $\ker A=D$, there is a unique matrix $B$ such that $BA=H$. Indeed, for each $\mathbf{z}\in\mathbb{F}_q^N$, choose $\mathbf{x}\in\mathbb{F}_q^n$ such that $A\mathbf{x}=\mathbf{z}$ and set $B\mathbf{z}:=H\mathbf{x}$. Since $\ker A\subseteq\ker H$, the value of $B\mathbf{z}$ does not depend on the choice of $\mathbf{x}$. This defines a linear map $B$ of rank $r$, since $BA=H$ and $A$ has rank $N$. The number of choices of $A$ is the same for every $D$, and the number of choices of $D$ is the same for every $H$ of rank $r$. Thus both claims follow.

Let $E$ be the event that $\ker H$ is not $(p,L)$ average radius list decodable. Suppose that $E$ occurs. By Lemma~\ref{lem:ar-reformulation}, there exist $L$ distinct vectors $\mathbf{x}_1,\ldots,\mathbf{x}_L\in\mathbb{F}_q^n$ such that $H\mathbf{x}_1=\cdots=H\mathbf{x}_L$ and $\prod_{i=1}^L\tau(\mathbf{x}_i)\ge q^{-Lnh_q(p)}$.

Fix these vectors for each such $H$ and let $\mathbf{z}_i:=A\mathbf{x}_i$ for $i\in[L]$. Since $H=BA$, we have $B\mathbf{z}_1=\cdots=B\mathbf{z}_L$. For any distinct $i,j\in[L]$, the vector $\mathbf{x}_i-\mathbf{x}_j$ is a fixed nonzero vector in $\ker H$ once $H$ is given. Given $H$, the kernel of $A$ is uniform among the subspaces of $\ker H$ of dimension $n-N$. Hence
\[
\Pr[\mathbf{z}_i=\mathbf{z}_j\mid H]
=\Pr[\mathbf{x}_i-\mathbf{x}_j\in\ker A\mid H]
=\frac{q^{n-N}-1}{q^{n-r}-1}\le q^{-t}.
\]
By a union bound, given any $H$ for which $E$ occurs, the vectors $\mathbf{z}_1,\ldots,\mathbf{z}_L$ are distinct with probability at least $1-\binom{L}{2}q^{-t}$.

By the definition of $A\tau$ in Section~\ref{sec:prelims}, $(A\tau)(\mathbf{z})=\sum_{\mathbf{x}:A\mathbf{x}=\mathbf{z}}\tau(\mathbf{x})$ for every $\mathbf{z}\in\mathbb{F}_q^N$. In particular, $(A\tau)(A\mathbf{x})\ge\tau(\mathbf{x})$ for every $\mathbf{x}\in\mathbb{F}_q^n$. For each $\mathbf{z}\in\mathbb{F}_q^N$, define $f(\mathbf{z}):=(1+q^r(A\tau)(\mathbf{z}))^{1/N}$ and let
\[
M(A,B):=\max_{\mathbf{b}\in\mathbb{F}_q^r}\sum_{\mathbf{z}:B\mathbf{z}=\mathbf{b}}(f(\mathbf{z})-1).
\]
For each fixed $A$, we have $f(\mathbf{z})\ge1$ for every $\mathbf{z}\in\mathbb{F}_q^N$ and $\sum_{\mathbf{z}\in\mathbb{F}_q^N}f(\mathbf{z})^N=q^N+q^r$. Since $A$ and $B$ are independent, Lemma~\ref{lem:ar-sum} gives $\mathbb{E}[M(A,B)\mid A]\le q-1+O_q(1/N+q^{-r})$. Taking the expectation over $A$, we obtain $\mathbb{E}_{A,B}[M(A,B)]\le q-1+o(1)$.

Now suppose that $E$ occurs and $\mathbf{z}_1,\ldots,\mathbf{z}_L$ are distinct. Since $(A\tau)(\mathbf{z}_i)\ge\tau(\mathbf{x}_i)$ for every $i\in[L]$, we have
\[
\prod_{i=1}^L(A\tau)(\mathbf{z}_i)\ge\prod_{i=1}^L\tau(\mathbf{x}_i)\ge q^{-Lnh_q(p)}.
\]
Since $B\mathbf{z}_1=\cdots=B\mathbf{z}_L$, $f(\mathbf{z})\ge1$ and $f(\mathbf{z})\ge q^{r/N}(A\tau)(\mathbf{z})^{1/N}$, the AM-GM inequality gives
\[
M(A,B)\ge\sum_{i=1}^L f(\mathbf{z}_i)-L\ge Lq^{r/N}\left(\prod_{i=1}^L(A\tau)(\mathbf{z}_i)\right)^{1/(LN)}-L\ge L(q^{(r-nh_q(p))/N}-1).
\]

Thus $M(A,B)\ge L(q^{(r-nh_q(p))/N}-1)$ with probability at least $\Pr[E](1-\binom{L}{2}q^{-t})$. Since $M(A,B)\ge0$, we obtain
\[
\Pr[E]\left(1-\binom{L}{2}q^{-t}\right)L(q^{(r-nh_q(p))/N}-1)\le\mathbb{E}_{A,B}[M(A,B)]\le q-1+o(1).
\]
Since $r=(h_q(p)+\varepsilon)n+O(1)$, $N=r+t$ and $t=\lfloor\sqrt n\rfloor$, we have $(r-nh_q(p))/N\to\varepsilon/(h_q(p)+\varepsilon)$ as $n\to\infty$. As $L$ is fixed, it follows that
\[
\Pr[E]\le\frac{q-1}{L(q^{\varepsilon/(h_q(p)+\varepsilon)}-1)}+o(1).
\]
By the assumption on $L$, there exists a constant $c=c(p,q,\varepsilon,L)>0$ such that $\Pr[E]\le1-c$ for all sufficiently large $n$.
\end{proof}

\begin{proof}[Proof of Theorem~\ref{thm:ar-upper}]
For $L$ as in the theorem, we have $L(q^{\varepsilon/(h_q(p)+\varepsilon)}-1)>q-1$. By continuity, we can choose a constant $0<\varepsilon'<\varepsilon$ such that $L(q^{\varepsilon'/(h_q(p)+\varepsilon')}-1)>q-1$. Let $k':=\lfloor(1-h_q(p)-\varepsilon')n\rfloor$. By Proposition~\ref{prop:ar-positive-probability}, a random linear code of dimension $k'$ is $(p,L)$ average radius list decodable with probability at least $c$ for some constant $0<c<1$.

By Lemma~\ref{lem:ar-threshold} and the fact that $k'-k=(\varepsilon-\varepsilon')n+O(1)$, the probability that $\cC$ is not $(p,L)$ average radius list decodable is at most $\frac{1-c}{c}q^{-(k'-k)}=q^{-\Omega_{p,q,\varepsilon}(n)}$.

Finally, for fixed $p,q$, as $\varepsilon\to0$ we have $q^{\varepsilon/(h_q(p)+\varepsilon)}-1=\frac{\varepsilon\ln q}{h_q(p)}+O_{p,q}(\varepsilon^2)$.
Hence $L=\frac{(q-1)h_q(p)}{\varepsilon\ln q}+O_{p,q}(1)$.
\end{proof}

\paragraph{Acknowledgment of AI assistance.}
GPT-6 assisted with the proofs in Section~\ref{sec:average-upper} and with language editing throughout the paper. All other results and proofs were obtained by the authors without AI assistance. The authors take full responsibility for the correctness of the results and the quality of exposition.


\bibliographystyle{alpha}
\bibliography{refs}

\appendix
\section{Proofs of Lemma~\ref{lem:uniform-common-shift}}
\label{app:proof-glmrsw-lemma45}

\begin{proof}[Proof of Lemma~\ref{lem:uniform-common-shift}]
If $c\le 0$, the assertion follows immediately from $H_q(\mathbf{v}+\beta\mathbf{w})\ge H_q(\mathbf{v}+\beta\mathbf{w}\mid \beta)=d\,h_q(p_1),$ so assume $0<c<1$. We will bound $H_q(\mathbf{v}+\beta\mathbf{w})$ in two cases, one when $d$ is small relative to $m$ and one when $d$ is large.

By multiplying the $i$-th coordinate by $w_i^{-1}$, we may assume without loss of generality that $\mathbf{w}=\mathbf{1}_d.$ Indeed, multiplication by a nonzero element of $\mathbb{F}_q$ preserves the distribution $\mathrm{Bernoulli}_q(p_1)$ and does not change entropy. Write $y_i:=v_i+\beta$ for $1\leq i\leq d$ and let $\mathbf{y}=(y_1,\ldots,y_d)$.

First we consider the case where $d$ is small. Since $\beta$ is uniform and independent of $\mathbf{v}$, $y_1=v_1+\beta$ is uniform over $\mathbb{F}_q$. Hence
\begin{align*}
H_q(\mathbf{y})
&=
H_q(y_1)
+
H_q(y_2,\ldots,y_d\mid y_1)\ge
H_q(y_1)
+
H_q(y_2,\ldots,y_d\mid y_1,\beta)\\
&=
1+(d-1)h_q(p_1)=
d\,h_q(p_1)+\bigl(1-h_q(p_1)\bigr).
\end{align*}
Recall that $p_1\to p$ as $m\to\infty$. 

Since $p_1\to p$ and $h_q(p)<1$, $1-h_q(p_1)\ge \frac{1-h_q(p)}{2}$ for all sufficiently large $m$. Then, whenever $d\le \min\left\{1,\frac{1-h_q(p)}{2c}\right\} m$, we have $H_q(\mathbf y)\ge d\,h_q(p_1)+1-h_q(p_1)\ge d\,h_q(p_1)+c\frac{d}{m}.$ This proves the desired bound when $d$ is small.

Next we handle the case where $d$ is large. Consider the invertible linear transformation $(y_1,y_2,\ldots,y_d)\longmapsto(y_1,z_2,\ldots,z_d),$ where $z_i:=y_i-y_1=v_i-v_1$ for $2\leq i\leq d$. 
Condition on $v_1=a$. For every $i\ge2$, $\Pr[z_i=-a\mid v_1=a]=1-p_1,$ while, for every $b\neq a$, $\Pr[z_i=-b\mid v_1=a]=\frac{p_1}{q-1}.$ As $1-p_1>\frac{p_1}{q-1}$ for all sufficiently large $m$, with a gap bounded below by a positive constant depending only on $p$ and $q$, the most frequent symbol among $z_2,\ldots,z_d$ can be used to recover $-v_1$ with high probability.

Let $v'$ be the most frequent symbol among $z_2,\ldots,z_d$. By Hoeffding's inequality, $\Pr[-v'\neq v_1]\le q\exp(-\Omega_{p,q}(d)).$ Therefore, by Fano's inequality, $H_q(v_1\mid z_2,\ldots,z_d)\le \exp(-\Omega_{p,q}(d)).$ Since $\mathbf v$ and $(v_1,z_2,\ldots,z_d)$ uniquely determine each other, we have $d\,h_q(p_1)=H_q(\mathbf v)=H_q(z_2,\ldots,z_d)+H_q(v_1\mid z_2,\ldots,z_d).$ Thus $H_q(z_2,\ldots,z_d)\ge d\,h_q(p_1)-\exp(-\Omega_{p,q}(d)).$

On the other hand, since $\beta$ is uniform and independent of $\mathbf{v}$, $H_q(y_1\mid z_2,\ldots,z_d)=1.$ Therefore,
\begin{align*}
H_q(\mathbf{y})
&=
H_q(y_1,z_2,\ldots,z_d)=
H_q(z_2,\ldots,z_d)
+
H_q(y_1\mid z_2,\ldots,z_d) \\
&\ge
d\,h_q(p_1)+1-\exp(-\Omega_{p,q}(d)).
\end{align*}
Since $c<1$, there exists a constant $d_0=d_0(p,q,c)$ such that, for every $d\ge d_0$, $1-\exp(-\Omega_{p,q}(d))\ge c.$ Thus, for $d\ge d_0$, $H_q(\mathbf{y})\ge d\,h_q(p_1)+c \ge d\,h_q(p_1)+c\frac{d}{m},$ where the last inequality uses $d\le m$.

Finally, take $m$ sufficiently large so that $d_0\le\min\left\{1,\frac{1-h_q(p)}{2c}\right\} m.$ Then every $1\le d\le m$ is covered by one of the two cases above: if $d\le\min\left\{1,\frac{1-h_q(p)}{2c}\right\} m$, we use the small $d$ bound; otherwise $d>\min\left\{1,\frac{1-h_q(p)}{2c}\right\} m\ge d_0$ and we use the large $d$ bound. This proves the lemma.
\end{proof}

\section{Proofs for Section~\ref{sec:average-multiline}}
\label{app:multiline}

We use the notation of Section~\ref{sec:average-multiline}. Throughout this appendix, $p,q,s$ are fixed. Whenever $\zeta$ and $\gamma$ are used, they are as in Lemma~\ref{ml:lem:not-rare}. We take $m$ sufficiently large depending only on $p,q,s,\zeta$.

\begin{proof}[Proof of Proposition~\ref{ml:prop:bad}]
Let $M\in\mathcal M_{n,\tau}$\footnote{ As per Remark~\ref{rem:type-rounding}, we may ignore the rounding issue that $\tau$ may not be realized exactly as an empirical row distribution. We round down the frequency of each nonzero row and fill the remaining positions with zero rows; this does not increase the average distance. } and let $M_1,\ldots,M_L$ denote its columns. We want to show that there exists a vector $\mathbf z\in\mathbb F_q^n$ such that $\frac1L\sum_{i=1}^L d(M_i,\mathbf z)\le pn$.

By the definition of $\tau$ and $\mathcal M_{n,\tau}$, we may write the $r$-th row of $M$ as
\[
M_{r,*}
=
\Big(
\alpha_r+E_{r,1},\ldots,\alpha_r+E_{r,m},
\alpha_r+B_r,
(\alpha_r+B_r+aT_{r,j})_{
1\le j\le s,\;a\in\mathbb F_q^*}
\Big).
\]
For each row, there may be several representations allowed by Definition~\ref{ml:def:distribution}. We choose one for which the row differs from $\alpha_r$ in as few coordinates as possible.

Set $\mathbf z:=(\alpha_1,\ldots,\alpha_n)\in\mathbb F_q^n$. For each of the first $m$ columns, the element in row $r$ differs from $\alpha_r$ exactly when $E_{r,i}\neq0$. Then we count the differences among the last $1+s(q-1)$ columns. For each $1\le j\le s$, consider the $q$ elements $(\alpha_r+B_r+aT_{r,j})_{a\in\mathbb F_q}$. If $T_{r,j}=0$, then $B_r=0$, and all $q$ elements are equal to $\alpha_r$. If $T_{r,j}\neq0$, then $a\mapsto B_r+aT_{r,j}$ is a permutation of $\mathbb F_q$, so exactly $q-1$ of these elements differ from $\alpha_r$.

Counting the $q$ elements for each $T_j$ separately gives $(q-1)\sum_{j=1}^s \mathbf 1[T_{r,j}\neq0]$. However, the element $\alpha_r+B_r$, corresponding to $a=0$, is counted once for each $T_j$ in the above sum but appears only once in $M$; when $B_r\neq0$, it is therefore overcounted by $s-1$. Hence, the number of differences among the last $1+s(q-1)$ columns in row $r$ is $(q-1)\sum_{j=1}^s \mathbf 1[T_{r,j}\neq0] -(s-1)\mathbf 1[B_r\neq0]$. Since $\Pr[B\neq0]=\left(1-\frac1q\right)p_2^s$, we obtain
\[
\begin{aligned}
\frac1L\sum_{i=1}^L d(M_i,\mathbf z)
&\le
\frac{n}{L}
\mathbb E\!\left[
\sum_{i=1}^m\mathbf 1[E_i\neq0]
+(q-1)\sum_{j=1}^s\mathbf 1[T_j\neq0]
-(s-1)\mathbf 1[B\neq0]
\right] \\
&=
\frac{n}{L}
\left(
mp_1+s(q-1)p_2
-(s-1)\left(1-\frac1q\right)p_2^s
\right) \\
&=
\frac{n}{L}
\left(
pL
-(s-1)\left(1-\frac1q\right)(p_2^s-p^s)
-\frac1m
\right) \\
&\le
pn-\frac{n}{mL}
< pn.
\end{aligned}
\]
For each row, the chosen representation has no more disagreements than the average under Definition~\ref{ml:def:distribution}, conditioned on producing that row. Since $M$ has type $\tau$, averaging over its rows gives the first inequality.

Finally, for any two coordinates $i\neq j$, there exists a vector in $\mathrm{supp}(\tau)$ whose $i$-th and $j$-th entries are different. Thus, for all sufficiently large $n$, the columns $M_1,\ldots,M_L$ are pairwise distinct. Therefore, they form a bad list of size $L$ in $\cC$, and $\cC$ is not $(p,L)$ average radius list decodable.
\end{proof}

\begin{proof}[Proof of Lemma~\ref{ml:lem:last-block}]
For all sufficiently large $m$, Taylor expansion gives $\gamma-h_q(p_1)\le(p^s-\zeta/2)/m$ while $h_q(p_2)\ge\gamma$ and $\gamma<1$.

Since $A$ has rank $j$, choose $j$ linearly independent columns of $A$ and let $T_{i_1},\ldots,T_{i_j}$ be the corresponding coordinates of $\mathbf T$. Once all the other coordinates of $\mathbf T$ are fixed, we have
\[
A\mathbf T=A'
\begin{pmatrix}T_{i_1}\\\vdots\\T_{i_j}\end{pmatrix}
+\mathbf c,
\]
where $A'\in\mathbb F_q^{j\times j}$ is invertible and $\mathbf c$ is fixed. Since $T_1,\ldots,T_s$ are independent,
\[
\begin{aligned}
H_q(A\mathbf T)
\ge H_q\!\left(A\mathbf T\,\middle|\,
(T_i)_{i\notin\{i_1,\ldots,i_j\}}\right)=H_q(T_{i_1},\ldots,T_{i_j})
=j h_q(p_2)\ge j\gamma.
\end{aligned}
\]

If $d=0$, it already proves the lemma, so assume $d>0$. Then we estimate $H_q((E_1,\ldots,E_d)-B\mathbf1_d\mid\mathbf T)$ as in the proof of Lemma~\ref{lem:not-implicitly-rare}. For any fixed value of $\mathbf T$, if at least one $T_i$ is zero, then $B=0$; if every $T_i$ is nonzero, then $-B$ is uniform over $\mathbb F_q$ and independent of $E_1,\ldots,E_d$. In the latter case, this conditional distribution does not depend on the nonzero values of $T_1,\ldots,T_s$. Thus, applying Lemma~\ref{lem:uniform-common-shift} with $c=1-\zeta/(4p^s)$ gives
\[
\begin{aligned}
H_q\!\left((E_1,\ldots,E_d)-B\mathbf1_d\,\middle|\,\mathbf T\right)
&=(1-p_2^s)d h_q(p_1)
+p_2^s H_q\!\left(
(E_1,\ldots,E_d)-B\mathbf1_d
\,\middle|\,T_1\ne0,\ldots,T_s\ne0\right)\\
&\ge(1-p_2^s)d h_q(p_1)
+p_2^s\left(d h_q(p_1)
+\left(1-\frac{\zeta}{4p^s}\right)\frac d m\right)\\
&=d h_q(p_1)+p_2^s\left(1-\frac{\zeta}{4p^s}\right)\frac d m\\
&\ge d h_q(p_1)+\left(p^s-\frac{\zeta}{4}\right)\frac d m
\ge d\gamma.
\end{aligned}
\]
where the last inequality follows from the above bound on $\gamma-h_q(p_1)$. Since $A\mathbf T$ and $Q\mathbf T$ are fixed once $\mathbf T$ is fixed, the chain rule gives
\[
\begin{aligned}
H_q\!\left(A\mathbf T,(E_1,\ldots,E_d)-B\mathbf1_d+Q\mathbf T\right)
&=H_q(A\mathbf T)
+H_q\!\left((E_1,\ldots,E_d)-B\mathbf1_d+Q\mathbf T
\,\middle|\,A\mathbf T\right)\\
&\ge H_q(A\mathbf T)
+H_q\!\left((E_1,\ldots,E_d)-B\mathbf1_d+Q\mathbf T
\,\middle|\,\mathbf T\right)\\
&=H_q(A\mathbf T)
+H_q\!\left((E_1,\ldots,E_d)-B\mathbf1_d
\,\middle|\,\mathbf T\right)\\
&\ge j\gamma+d h_q(p_1)
+\left(p^s-\frac{\zeta}{4}\right)\frac d m
\ge(j+d)\gamma.
\end{aligned}
\]
This proves the lemma.
\end{proof}

\begin{proof}[Proof of Lemma~\ref{ml:lem:not-rare}]
We first show that every rank-$r$ linear map $A:\mathbb F_q^{m+s}\to\mathbb F_q^r$ satisfies $H_q(A(\mathbf U,\mathbf T))\ge r\gamma$. The case $r=0$ is immediate. Write $A=[P\mid Q]$ and define $\mathbf Y:=P\mathbf U+Q\mathbf T$.

\medskip
\noindent\textbf{Case 1: $\mathrm{rank}(P)=r$.}

We prove the bound $H_q(\mathbf Y\mid\mathbf T)\ge r\gamma$, following Case~1 in the proof of Lemma~\ref{lem:not-implicitly-rare}. Permuting the coordinates of $\mathbf U$ does not change their joint distribution with $\mathbf T$ and invertible row operations preserve entropy. Thus, we may assume $P=[I_r\mid\mathbf w^{(1)}\mid\cdots\mid\mathbf w^{(k)}]$ with $k=m-r$. Let $\mathbf V:=(E_1,\ldots,E_r)$ and $G_j:=E_{r+j}$ for $1\le j\le k$. Then
\[
\mathbf Y=\mathbf V+\sum_{j=1}^k\mathbf w^{(j)}G_j
+\mathbf bB+Q\mathbf T,
\qquad
\mathbf b:=-\mathbf1_r-\sum_{j=1}^k\mathbf w^{(j)}.
\]
For $1\le j\le k$, define $J_j:=\mathrm{supp}(\mathbf w^{(j)})\setminus\bigcup_{a<j}J_a$ and $J_{k+1}:=[r]\setminus\bigcup_{j=1}^kJ_j$.

Once $\mathbf V_{J_{j+1}},\ldots,\mathbf V_{J_{k+1}}, G_{j+1},\ldots,G_k,B,T$ are fixed, $\mathbf Y_{J_{j+1}},\ldots,\mathbf Y_{J_{k+1}}$ are fixed as well, and $\mathbf Y_{J_j}=\mathbf V_{J_j}+\mathbf w^{(j)}_{J_j}G_j+\mathbf c$ for a fixed $\mathbf c$. The same conditional entropy argument as in Case~1 of Lemma~\ref{lem:not-implicitly-rare} therefore gives $H_q\!\left(\mathbf Y_{J_j}\,\middle|\,\mathbf Y_{J_{j+1}},\ldots,\mathbf Y_{J_{k+1}},\mathbf T\right)\ge H_q\!\left(\mathbf V_{J_j}+\mathbf w^{(j)}_{J_j}G_j\right).$ Moreover, $\mathbf Y_{J_{k+1}}=\mathbf V_{J_{k+1}}-B\mathbf1_{|J_{k+1}|}+(Q\mathbf T)_{J_{k+1}}.$ Applying the chain rule to $\mathbf Y_{J_{k+1}},\ldots,\mathbf Y_{J_1}$ gives
\[
\begin{aligned}
H_q(\mathbf Y\mid\mathbf T)
&=H_q(\mathbf Y_{J_{k+1}}\mid\mathbf T)
+\sum_{j=1}^k H_q\!\left(\mathbf Y_{J_j}\,\middle|\,
\mathbf Y_{J_{j+1}},\ldots,\mathbf Y_{J_{k+1}},\mathbf T\right)\\
&\ge H_q\!\left(\mathbf V_{J_{k+1}}-B\mathbf1_{|J_{k+1}|}
\,\middle|\,\mathbf T\right)
+\sum_{j=1}^k H_q\!\left(\mathbf V_{J_j}+\mathbf w^{(j)}_{J_j}G_j\right).
\end{aligned}
\]
For all sufficiently large $m$, we have $\gamma-h_q(p_1)\le(p^s-\zeta/2)/m$. For every $1\le j\le k$, Corollary~\ref{cor:ordinary-shift} applied with $c=p^s-\zeta/4<h_q(p)$ gives $H_q\!\left(\mathbf V_{J_j}+\mathbf w^{(j)}_{J_j}G_j\right)\ge |J_j|h_q(p_1)+\left(p^s-\frac{\zeta}{4}\right)\frac{|J_j|}{m}\ge |J_j|\gamma.$ For $\mathbf Y_{J_{k+1}}$, conditioning on $\mathbf T$ fixes the $Q\mathbf T$ term. Hence, the conditional entropy estimate proved in Lemma~\ref{ml:lem:last-block} gives
\[
H_q\!\left(
\mathbf V_{J_{k+1}}
-B\mathbf 1_{|J_{k+1}|}
\,\middle|\,\mathbf T
\right)
\ge
|J_{k+1}|h_q(p_1)
+
\left(p^s-\frac{\zeta}{4}\right)
\frac{|J_{k+1}|}{m}
\ge
|J_{k+1}|\gamma.
\]
Combining these bounds, we obtain $H_q(\mathbf Y\mid\mathbf T)\ge\gamma\sum_{j=1}^{k+1}|J_j|=r\gamma.$ Thus, $H_q(\mathbf Y)\ge r\gamma$.

\medskip
\noindent\textbf{Case 2: $\mathrm{rank}(P)=r-j$ for some $1\le j\le\min\{r,s\}$.}

Since $P$ has rank $r-j$, invertible row operations make its last $j$ rows zero. Applying the same operations to $A=[P\mid Q]$ gives
\[
\begin{pmatrix}
P_0 & Q_0\\
0 & A_0
\end{pmatrix},
\qquad
\mathrm{rank}(P_0)=r-j.
\]
The full matrix still has rank $r$, the last $j$ rows are linearly independent and $\mathrm{rank}(A_0)=j$. These row operations preserve entropy. Since $A_0\mathbf T$ and $Q_0\mathbf T$ are fixed once $\mathbf T$ is fixed,
\[
\begin{aligned}
H_q(A(\mathbf U,\mathbf T))
&=H_q(A_0\mathbf T)+H_q(P_0\mathbf U+Q_0\mathbf T\mid A_0\mathbf T)\ge H_q(A_0\mathbf T)+H_q(P_0\mathbf U+Q_0\mathbf T\mid\mathbf T)\\
&=H_q(A_0\mathbf T)+H_q(P_0\mathbf U\mid\mathbf T)\ge j\gamma+(r-j)\gamma=r\gamma
\end{aligned}
\]
where the last inequality uses $H_q(A_0\mathbf T)\ge jh_q(p_2)\ge j\gamma$, as shown in the proof of Lemma~\ref{ml:lem:last-block} and $H_q(P_0\mathbf U\mid\mathbf T)\ge (r-j)\gamma$ from Case~1.

Finally, we add back $Z$ which is uniform and independent of $(\mathbf U,\mathbf T)$. Since $\gamma<1$ for sufficiently large $m$, the same argument as in the last part of the proof of Lemma~\ref{lem:not-implicitly-rare} gives $H_q(A'\mathbf W)\ge \gamma\,\mathrm{rank}(A')$ for every linear map $A':\mathbb F_q^{m+s+1}\to\mathbb F_q^r.$ For every full rank linear map $A:\mathbb F_q^L\to\mathbb F_q^{L'}$, there exists a linear map $A':\mathbb F_q^{m+s+1}\to\mathbb F_q^{L'}$ such that $A\mathbf X=A'\mathbf W$ and $\dim(A\tau)=\mathrm{rank}(A').$ Therefore, $H_q(A\tau)=H_q(A'\mathbf W)\ge\gamma\,\mathrm{rank}(A')=\gamma\,\dim(A\tau).$ Hence, $\tau$ is not $\gamma$-implicitly rare.
\end{proof}

\begin{proof}[Proof of Theorem~\ref{ml:thm:lower bound}]

We follow the proof of Theorem~\ref{thm:average radius-lower bound} after replacing $p(1+h_q'(p))$ by $K_s$ and taking $L=m+1+s(q-1)$.

Fix $\delta>0$. Choose $0<\zeta<\min\{\delta/2,p^s/4,K_s/4\}$ and set $m:=\left\lfloor(K_s-2\zeta)/\varepsilon\right\rfloor$. For sufficiently small $\varepsilon$, Lemma~\ref{ml:lem:not-rare} shows that $\tau$ is not $\gamma$-implicitly rare where $\gamma=h_q(p)+(K_s-\zeta)/m$. Since $\gamma-(1-R)=\frac{K_s-\zeta}{m}-\varepsilon\ge\frac{\zeta}{m},$ $\tau$ is not $(1-R+\zeta/(2m))$-implicitly rare.

By Theorem~\ref{thm:abundance}, a random linear code $\cC$ of rate $R$ contains some matrix $M\in\mathcal M_{n,\tau}$ with probability at least $1-n^{O_{L,q}(1)}q^{-\zeta n/(2m)}=1-q^{-\Omega_{p,q,s,\delta,\varepsilon}(n)}$. By Proposition~\ref{ml:prop:bad}, the $L=m+1+s(q-1)$ columns of $M$ form a bad list for $(p,L)$ average radius list-decoding.

Finally, since $2\zeta<\delta$, we have $L=m+1+s(q-1)\ge\frac{K_s-2\zeta}{\varepsilon}+s(q-1)\ge\frac{K_s-\delta}{\varepsilon}.$ Therefore, $\cC$ is not $\left(p,\left\lfloor\frac{K_s-\delta}{\varepsilon}\right\rfloor\right)$ average radius list decodable, which proves the theorem.
\end{proof}

\end{document}